\documentclass[11pt]{article}
 \usepackage{geometry}                % See geometry.pdf to learn the layout options. There are lots.
\usepackage{graphicx}
\usepackage{amssymb}
\usepackage{epstopdf}
\usepackage{graphicx}
\usepackage{amssymb}
\usepackage{amsmath}
\usepackage{amsfonts}
\usepackage{amsthm}  
\usepackage{mathabx}
\usepackage{enumitem}
\usepackage{subfigure}
\usepackage{algpseudocode} 
\usepackage{algorithm}

\usepackage{cite}
\usepackage{subfloat}

\newtheorem{theorem}{Theorem}[section]
\newtheorem{remark}[theorem]{Remark}
\newtheorem{lemma}[theorem]{Lemma}
\newtheorem{corollary}[theorem]{Corollary}

\newcommand{\im}{\mathrm{i}}

\newcommand{\dif}{\mathrm{d}}

\renewcommand{\vec}[1]{\mbox{\boldmath$#1$}}

\begin{document}
\begin{center}
{\bf \Large Understanding the Magnetic Polarizability Tensor}\\[0.3cm]
P.D. Ledger$^*$ and W.R.B. Lionheart$^\dagger$\\[0.3cm]
$^*$P.d.ledger@swansea.ac.uk, $^\dagger$Bill.Lionheart@manchester.ac.uk\\[0.3cm]
$^*$College of Engineering, Swansea
University Bay Campus, \\
Fabian Way, Crymlyn Burrows, SA1 8EN UK\\[0.3cm]
$^\dagger$School of Mathematics, The University of
Manchester,\\
Alan Turing Building, Oxford Road, Manchester, M13 9PL UK\\[0.3cm]
2nd October 2015
\end{center}

\section*{Abstract}
The aim of this paper is provide new insights into the properties of the rank 2 polarizability tensor $\widecheck{\widecheck{\mathcal M}}$ proposed in (P.D. Ledger and W.R.B. Lionheart {\em Characterising the shape and material properties of hidden targets from magnetic induction data}, IMA Journal of Applied Mathematics, doi: 10.1093/imamat/hxv015)  for describing the perturbation in the magnetic field caused by the presence of a conducting object in the eddy current regime. In particular, we explore its connection with the
magnetic polarizability tensor and
 the P\'olya--Szeg\"o tensor and how, by introducing new splittings of $\widecheck{\widecheck{\mathcal M}}$, they form a family of rank 2 tensors for describing the response from different categories of conducting (permeable)  objects. We include new bounds on the invariants of the P\'olya--Szeg\"o tensor and expressions for the low frequency and high conductivity limiting coefficients of $\widecheck{\widecheck{\mathcal M}}$. We show, for the high conductivity case (and for frequencies at the limit of the quasi-static approximation), that it is important to consider whether the object is simply or multiply connected but, for the low frequency case, the coefficients are independent of the connectedness of the object. Furthermore, we explore the frequency response  of the coefficients of $\widecheck{\widecheck{\mathcal M}}$ for a range of simply and multiply connected objects.

\vspace{0.2cm}

\noindent {\bf Keywords:} Metal detectors, Land mine detection, Polarizability tensors, Eddy currents, Magnetic induction

\section{Introduction}

The need to detect and characterise conducting targets from magnetic induction measurements arises in a wide range of applications, most notably in metal detection.
Here one wishes to be able to locate and identify a highly conducting object in a low conducting background. Applications include ensuring safety at airports and at public events, maintaining quality in the mechanised production of food as well as in the detection of unexploded ordnance and land mines and in archaeological surveys. Furthermore, there is interest in producing conductivity images from multiple magnetic induction measurements, most notably in magnetic induction tomography for medical applications~\cite{huw2001,zolgharni2009} and industrial applications~\cite{soleimani2007, gaydecki2000}. Eddy currents also have important applications in the non--destructive testing, such as investigating the integrity of reinforced concrete structures and bridges~\cite{simmonds1999}.

In the engineering literature, the signal induced from an alternating low--frequency magnetic field, due to the presence of a conducting (permeable) object, located at the position ${\vec z}$,
 is often postulated as (e.g.~\cite{norton2001,das1990,peyton2013,pasion2001})
%%%%%%%%%%%%%%%%%%%%%%%%%%%%%%%%%%%%%%%%%%%%%%%%%%%%%%%%%%%%%%%%%%%%%%%%%%%%%%%
\begin{equation} 
V^{ind} \approx {\vec H}_0^m \cdot  ({\mathcal A} {\vec H}_0^e ), \label{eqn:engexp}
\end{equation}
%%%%%%%%%%%%%%%%%%%%%%%%%%%%%%%%%%%%%%%%%%%%%%%%%%%%%%%%%%%%%%%%%%%%%%%%%%%%%%%
where only ${\vec H}_0^m$ and ${\vec H}_0^e$ depend on position, ${\vec H}_0^e:={\vec H}_0({\vec z})$ is the background magnetic field generated by passing a current though an excitor coil placed away from the object, evaluated at the position of the centre of the object, and ${\vec H}_0^m$ is the corresponding field that would be evaluated by passing a unit current through the measurement coil, evaluated at the same location. The object
%%%%%%%%%%%%%%%%%%%%%%%%%%%%%%%%%%%%%%%%%%%%%%%%%%%%%%%%%%%%%%%%%%%%%%%%%%%%%%%
\begin{equation}
\displaystyle {\mathcal A}= \sum_{i=1}^3 \sum_{j=1}^3 {\mathcal A}_{ij} \hat{\vec e}_i \otimes \hat{\vec e}_j,
\end{equation}
%%%%%%%%%%%%%%%%%%%%%%%%%%%%%%%%%%%%%%%%%%%%%%%%%%%%%%%%%%%%%%%%%%%%%%%%%%%%%%%
has been proposed to be a (complex) symmetric rank 2 magnetic polarizability tensor described by 6 independent empirically fitted complex coefficients ${\mathcal A}_{ij}$ containing information about the shape, material properties and frequency response of an object. The coefficients are independent of its position and $\hat{\vec e}_i$, $i=1,2,3$ are the unit basis vectors for the chosen coordinate system.

In the applied mathematics literature, asymptotic formulae are available that described the perturbation in the magnetic field due to the presence of a magnetic (conducting) object $B_\alpha$ using Einstein's summation convention in the form (e.g.~\cite{kleinman1967,kleinman1973,ammarikangbook,ammarivolkov2005,ammarivolkov2013,ammarivogeluisvolkov,ledgerlionheart2014})
%%%%%%%%%%%%%%%%%%%%%%%%%%%%%%%%%%%%%%%%%%%%%%%%%%%%%%%%%%%%%%%%%%%%%%%%%%%%%%%
\begin{align}  
(({\vec H}_\alpha - {\vec H}_0)({\vec x}))_i = & ({\vec D}_{\vec x}^2G({\vec x},{\vec z}))_{ij}  {\mathcal A} _{jk} ({\vec H}_0({\vec z}))_k 
                                                                      + ({\vec R}({\vec x}))_i ,\label{eqn:asympformula} 
\end{align}
%%%%%%%%%%%%%%%%%%%%%%%%%%%%%%%%%%%%%%%%%%%%%%%%%%%%%%%%%%%%%%%%%%%%%%%%%%%%%%%
as some suitable limit is taken. In the above, $B_\alpha = \alpha B + {\vec z}$, which means that the physical object can be expressed in terms of a unit object $B$ placed at the origin, scaled by the object size $\alpha$ and translated by the vector ${\vec z}$,  ${\vec R} ( {\vec x} )$ is a residual vector, $G({\vec x},{\vec z}) = 1/( 4 \pi | {\vec x}- {\vec z}|)$ is the free space Laplace Green's function and
%%%%%%%%%%%%%%%%%%%%%%%%%%%%%%%%%%%%%%%%%%%%%%%%%%%%%%%%%%%%%%%%%%%%%%%%%%%%%%%
\begin{align}
( {\vec D}_{\vec x}^2G({\vec x},{\vec z} ) )_{ij} &= ({\vec D}_{\vec x}^2G({\vec z},{\vec x} ))_{ij}  = \frac{1}{4 \pi  r^3 } \left (3  \hat{\vec r}_i  \hat{\vec r}_j - \delta_{ij} \right ) ,
\end{align}
%%%%%%%%%%%%%%%%%%%%%%%%%%%%%%%%%%%%%%%%%%%%%%%%%%%%%%%%%%%%%%%%%%%%%%%%%%%%%%%
with ${\vec r}: = {\vec x}- {\vec z}$, $r=|{\vec r}| $, $\hat{\vec r}= {\vec r} / r$ and $\delta_{ij}$ the coefficients of the unit tensor. For these asymptotic formulae the form of ${\mathcal A} _{jk}$ is explicitly known and is computable by solving a transmission problem.  It is easily established that this is exactly of the form   ${\vec H}_0^m \cdot  ({\mathcal A} {\vec H}_0^e )$ given in equation (\ref{eqn:engexp}). This is done by idealising the background magnetic field as that produced by a magnetic dipole 
and then
 taking the component of (\ref{eqn:asympformula}) in the direction of the magnetic moment associated with the background magnetic field that would result from the measurement coil being treated as an excitor.

In this paper, we will explore the connection between the empirically fitted engineering polarizability tensor and the asymptotic expansions combined with the different expressions for polarization tensor that appear in the applied mathematics literature. In particular, we advocate that an asymptotic formula for the perturbed magnetic field provides greater insight than (\ref{eqn:engexp}) for the following reasons. If the empirical approach is followed, the coefficients of the polarizability tensor are obtained by taking measurements (or performing simulated measurements using a computational techniques such as finite elements) in order to determine the voltage at different positions for different excitor combinations and then use a least squares approach to approximately determine the coefficients ${\mathcal A}_{ij}$ e.g.~\cite{das1990, norton2001,pasion2001,peyton2013}. It is important to ensure that measurements are taken in different planes and at distances from the object in order to capture the correct asymptotic behaviour, however, the presence of measurement noise can make this challenging to achieve in practice.
Furthermore, the number of measurements should greatly exceed the number of coefficients to be determined in order to minimise any measurement errors. For each new object, i.e. a different shape, frequency or material properties, this measurement procedure must be repeated in order to determine the polarizability tensor. An asymptotic formula, on the other hand, provides an explicit expression that allows the tensor to be computed {\em without} the need for performing (or simulating) measurements. Indeed, perhaps a contributing factor as to why this approach has not been persued  amongst engineers so far is that the term polarization tensor is preferred in the applied mathematics literature for ${\mathcal A}$ while, in the engineering community,  the term polarizability tensor is more commonly adopted.
 Another benefit is that rather than knowing that the voltage is only approximately given (\ref{eqn:engexp}), without knowing its accuracy, we have, through (\ref{eqn:asympformula}), not only the leading term but, also, a way of rigorously describing the remainder. 

 The availability of explicit expressions for different classes of polarization/ polarizability tensors makes it possible to investigate their properties, such as the reduction in the number of independent coefficients for rotational and reflectional symmetries of the object, which we considered in~\cite{ledgerlionheart2014}. In this work, we provide the following novel contributions, which further enhance the understanding of their properties: Firstly, we review the different forms that the tensors can take for magnetic and conducting (permeable) objects as part of
asymptotic expansions of the perturbed magnetic field when either the object size or the frequency tend to zero. Secondly, we present new bounds on the invariants of some classes of the tensors and provide bounds on the spherical and deviatoric parts of the tensor for a magnetic object. Thirdly, we present new results that describe the low frequency and high conductivity limits of the coefficients of the tensor for a conducting (permeable) object in an alternating background magnetic field. 
 We consider the response from magnetic and conducting ellipsoids and, finally, the response from a conducting Remington rifle cartridge as a practical application of the aforementioned theory.
 
The presentation of the material is organised as follows: In Section~\ref{ref:eddy}, we summarise the low frequency and eddy current models and then, in Section~\ref{sect:expexp}, we consider explicit expressions for the polarization/polarizability tensors for magnetic and conducting  (permeable) objects and, in the latter case, present a new splitting of the tensor. In Section~\ref{sect:bounds}, we present bounds on properties of the tensors and then, in Section~\ref{sect:lowhigh}, we consider the limiting case of low frequency and high conductivity for the coefficients of the tensor. Section~\ref{sect:ellip} describes the response from magnetic and conducting ellipsoids and, in Section~\ref{sect:rem}, the response from a conducting Remington rifle cartridge as a practical application of the aforementioned theory. We conclude the presentation with some concluding remarks in Section~\ref{sect:concl}.

 \section{Mathematical models} \label{ref:eddy}

Following Ammari {\it et al.}~\cite{ammarivolkov2013} we let ${\mathbb R}^3$ denote the Euclidean space and introduce the position dependent material parameters as
%%%%%%%%%%%%%%%%%%%%%%%%%%%%%%%%%%%%%%%%%%%%%%%%%%%%%%%%%%%%%%%%%%%%%%%%%%%%%%%
\begin{equation}
\epsilon_\alpha=\left \{ \begin{array}{l}
\epsilon_*  \\
\epsilon_0 \end{array} \right . , \,
\mu_\alpha=\left \{ \begin{array}{l}
\mu_*  \\
\mu_0 \end{array} \right . , \, 
\sigma_\alpha=\left \{ \begin{array}{ll}
\sigma_*  & \hbox{in $B_\alpha$}  \\
0  &  \hbox{in ${\mathbb R}^3 \setminus B_\alpha$} \end{array} \right . , 
\label{eqn:matparam}
\end{equation}
%%%%%%%%%%%%%%%%%%%%%%%%%%%%%%%%%%%%%%%%%%%%%%%%%%%%%%%%%%%%%%%%%%%%%%%%%%%%%%%
where $\epsilon$, $\mu$ and $\sigma$ are the permittivity, permeability and conductivity, respectively, and the subscript $0$ refers in the former cases to the free space values.
We remark that the background medium is assumed to be  non--conducting free space, which is a reasonable approximation to make for buried objects in dry ground provided that the contrast between the object and the surrounding material is sufficiently high. 

Low frequency electromagnetic scattering problems are described in terms of the (total) time harmonic ${\vec E}_\alpha$ and ${\vec H}_\alpha$ for angular frequency $\omega$, which result from the interaction between background (incident) fields ${\vec E}_0$ and ${\vec H}_0$ and the object $B_\alpha$. These fields satisfy the equations
\begin{subequations} \label{eqn:mathmodel1}
\begin{align}
\nabla \times {\vec E}_\alpha & = \im \omega \mu_\alpha {\vec H}_\alpha && \hbox{in ${\mathbb R}^3$} ,\\
\nabla \times {\vec H}_\alpha & = \sigma_\alpha {\vec E}_\alpha - \im \omega \epsilon_\alpha{\vec E}_\alpha 
&& \hbox{in ${\mathbb R}^3$} ,
\end{align} 
\end{subequations}
and suitable radiation conditions as $|{\vec x} |  \to \infty$. The background fields ${\vec E}_0$ and ${\vec H}_0$ satisfy the free space version of the above equations (i.e. replacing the subscript $\alpha$ with $0$). 

In the eddy current model, the geometry, frequency and material parameters are such that the 
 displacement currents in the Maxwell system  can be neglected. This is often justified on the basis that
$\sqrt{\epsilon_* \mu_* }\alpha \omega \ll 1 $ or $\epsilon_* \omega / \sigma_* \ll 1$.  A more rigorous justification of the eddy current model appears in~\cite{ammaribuffa2000}. In~\cite{schmidt2008} the effect of the shape of the conductor on the validity of the eddy current model is discussed. 
 The depth of penetration of the magnetic field in a conducting object is described by its skin depth, $s:=\sqrt{2/ (\omega \mu_0 \sigma_*) }$, and, by introducing a parameter $\nu := 2\alpha^2 / s^2$, the mathematical model of interest in this case refers to when $\nu =O(1)$~\footnote{Where $f(x)=O(g(x))$ if and only if there is a positive constant $M$ such that $|f(x) |  < M |g(x)  |$ for all sufficiently large $x$, ie $x \ge x_0$. This is known as Landau big O notation.} and $ \mu_*/ \mu_0 =O(1) $ as $\alpha \to 0$~\cite{ammarivolkov2013}.  When considering the eddy current model, the time harmonic fields ${\vec E}_\alpha$ and ${\vec H}_\alpha$ are those that result from a time varying current source located away from $B_\alpha$, with volume current density ${\vec J}_0$ and $\nabla\cdot{\vec J}_0=0$ in ${\mathbb R}^3$, and their interaction with the object $B_\alpha$, and satisfy the equations
%%%%%%%%%%%%%%%%%%%%%%%%%%%%%%%%%%%%%%%%%%%%%%%%%%%%%%%%%%%%%%%%%%%%%%%%%%%%%%%
\begin{subequations} \label{eqn:mathmodel2}
\begin{align} 
\nabla \times {\vec E}_\alpha & = \im \omega \mu_\alpha {\vec H}_\alpha && \hbox{in ${\mathbb R}^3$} ,\\
\nabla \times {\vec H}_\alpha & = \sigma_\alpha {\vec E}_\alpha +{\vec J}_0
&& \hbox{in ${\mathbb R}^3$} ,
\end{align} 
\end{subequations}
%%%%%%%%%%%%%%%%%%%%%%%%%%%%%%%%%%%%%%%%%%%%%%%%%%%%%%%%%%%%%%%%%%%%%%%%%%%%%%%
together with a suitable static decay rate of the fields as $|{\vec x}| \to \infty$~\cite{ammaribuffa2000}. In this case, in absence of an object, the background magnetic field, ${\vec H}_0$, is that generated by the current source.

 \section{Asymptotic formulae and explicit expressions for polarization tensors}\label{sect:expexp}
In Section~\ref{sect:mathmodel1} we discuss as asymptotic expressions for the perturbed magnetic field $({\vec H}_\alpha-{\vec H}_0)({\vec x})$ for the mathematical model described by (\ref{eqn:mathmodel1}). Then, in Section~\ref{sect:mathmodel2} we present expansions for the mathematical model described by (\ref{eqn:mathmodel2}). In Section~\ref{sect:unimathmodel3}, we consider the connection between these results.
 
 \subsection{Asymptotic expansions for the equation system (\ref{eqn:mathmodel1})} \label{sect:mathmodel1}
 Early results by Kleinman~\cite{kleinman1967,kleinman1973} presented the leading order terms for the perturbed (scattered fields) in terms of dipole moments for the case when $\kappa \to 0$ and $r \to \infty$, where $r$ is the distance from the object to the point of observation and $\kappa : =\omega \sqrt{\epsilon_0 \mu_0}$ is the free space wave number.  
Later it was shown how the moments could be expressed in terms of the dielectric and magnetic polarization/polarizability tensors multiplied by the incident electric and magnetic fields, respectively, evaluated at the position of the centre of the object~\cite{kleinman1972,kleinmansenior,kleinmanbook}. Our recent work~\cite{ledgerlionheart2012} includes not only terms as $r \to \infty$, but also includes those at distances that are large compared to the size of the object. In particular, by considering the case of  fixed $r$ and $\alpha$, this reduces to
%%%%%%%%%%%%%%%%%%%%%%%%%%%%%%%%%%%%%%%%%%%%%%%%%%%%%%%%%%%%%%%%%%%%%%%%%%%%%%%
\begin{align}  
(({\vec H}_\alpha - {\vec H}_0)({\vec x}))_i = &  ( {\vec D}_{\vec x}^2G({\vec x},{\vec z}))_{ij}  {\mathcal T}({\mu}_r)_{jk}  ( {\vec H}_0({\vec z}))_k                                                                 + {\vec R}({\vec x}) , \label{eqn:asympformula1} 
\end{align}
%%%%%%%%%%%%%%%%%%%%%%%%%%%%%%%%%%%%%%%%%%%%%%%%%%%%%%%%%%%%%%%%%%%%%%%%%%%%%%%
where ${\vec R}({\vec x})= O(\kappa ) $ as $\kappa \to 0$ for a simply connected smooth inclusion with permeability contrast ${\mu}_r: =\mu_* / \mu_0 $ at points away from the object, which describes the magnetostatic response as the limiting case of a low-frequency scattering problem. If we do not fix $r$ and $\alpha$, our Theorem 4.2~\cite{ledgerlionheart2012} can also describe the scattering from dielectric objects at distances that are large compared to the object size, but not the response from conducting objects. The real symmetric rank 2 tensor ${\mathcal T}({\mu}_r)$ can be expressed in range of different forms e.g.~\cite{kleinman1972,kleinmansenior,kleinmanbook} and can in fact be identified with the 
 P\'olya--Szeg\"o tensor~\cite{ammarikangbook}, whose coefficients are explicitly given by
%%%%%%%%%%%%%%%%%%%%%%%%%%%%%%%%%%%%%%%%%%%%%%%%%%%%%%%%%%%%%%%%%%%%%%%%%%%%%%%
\begin{align}
{\mathcal T}({\mu}_r )_{ij} =  & \alpha^3 \left (({\mu}_r -1) |B| \delta_{ij}   
+ ({\mu}_r-1)^2 \int_\Gamma \hat{\vec n}^- \cdot (\nabla_{\vec \xi} \phi_i)  \xi_j \dif {\vec \xi}\right ),
\end{align}
%%%%%%%%%%%%%%%%%%%%%%%%%%%%%%%%%%%%%%%%%%%%%%%%%%%%%%%%%%%%%%%%%%%%%%%%%%%%%%%
where $\xi_j=( {\vec \xi} )_j $ with ${\vec \xi}$ measured from the centre of $B$. In the above, $\phi_i$, $i=1,2,3$, satisfies the transmission problem
%%%%%%%%%%%%%%%%%%%%%%%%%%%%%%%%%%%%%%%%%%%%%%%%%%%%%%%%%%%%%%%%%%%%%%%%%%%%%%%
\begin{subequations}  \label{eqn:transmisscal}
\begin{align}
\nabla^2 \phi_i & =0 && \hbox{in $B\cup B^c$} ,\\
\left [ \phi_i \right ]_\Gamma & = 0 &&  \hbox{on $\Gamma$} , \\
\left . \hat{\vec n} \cdot \nabla \phi_i \right |_+ -\left . \hat{\vec n} \cdot \nabla {\mu}_r
 \phi_i \right |_- & =   \hat{\vec n} \cdot  \nabla \xi_i &&  \hbox{on $\Gamma$} , \\
\phi_i \to 0 & {}  &&  \hbox{as $| {\vec \xi}| \to \infty$}.
\end{align} 
\end{subequations}
%%%%%%%%%%%%%%%%%%%%%%%%%%%%%%%%%%%%%%%%%%%%%%%%%%%%%%%%%%%%%%%%%%%%%%%%%%%%%%%
where $\Gamma$ is the interface between $B$ and $B^c$ and $[\cdot]_\Gamma$ denotes the jump across $\Gamma$. Here, and in the sequel, we have dropped the subscript ${\vec \xi}$ on $\nabla$ and ${\vec x}$ on ${\vec D}$ unless confusion may arise.
Note that other forms of ${\mathcal T}({\mu}_r )_{ij}$ are possible and,
in~\cite{ledgerlionheart2012}, we show equivalence of some common arrangements.
By taking appropriate limiting
values of $ {\mu}_r$, the far field perturbation caused by the presence of  a perfectly conducting object can also be described~\cite{kleinmansenior,kleinmanbook}.

By contrast, the asymptotic behaviour of scattering by a small smooth simply connected object has been obtained on bounded~\cite{ammarivogeluisvolkov,franciswatsonbamc} and unbounded domains~\cite{ammarivolkov2005,ledgerlionheart2012} and, for the limiting magnetostatic response given by $\kappa=0$, has a similar form to 
 (\ref{eqn:asympformula1}) with $ {\vec R}({\vec x})= O(\alpha^3 ) $ as $\alpha \to 0$. 
 When $\kappa\ne 0$ these results also  include additional terms, which also describe the scattering from a conducting dielectric object in terms of ${\mathcal T}({\epsilon}_r^c)$ where ${\epsilon}_r^c := ( \epsilon_* - \im \sigma_* /  \omega )/\epsilon_0 $.

%%%%%%%%%%%%%%%%%%%%%%%%%%%%%%%%%%%%%%%%%%%%%%%%%%%%%%%%%%%%%%%%%%%%%%%%%%%%%%%
\subsection{Asymptotic expansions for the equation system (\ref{eqn:mathmodel2})} \label{sect:mathmodel2}
%%%%%%%%%%%%%%%%%%%%%%%%%%%%%%%%%%%%%%%%%%%%%%%%%%%%%%%%%%%%%%%%%%%%%%%%%%%%%%%

In the above, we have described a series of different asymptotic expansions, which although having a similar form to (\ref{eqn:asympformula}), and describing the perturbed fields for smooth simply connected magnetic objects in the magnetostatic regime (as limiting cases of electromagnetic scattering problems), do not describe the response from conducting objects in the quasi-static regime, i.e. the eddy current problem. The existence of a relationship of the form (\ref{eqn:engexp}) for the case of a conducting sphere in the eddy current regime was first shown by Wait~\cite{wait1951}. For this case ${\mathcal A}$ is diagonal and, he claimed that such a relationship could potentially be used for identifying the conductivity and radius of a sphere. Landau and Lifschitz~\cite[p. 192]{landau} propose that the total magnetic moment acquired by the conductor in a magnetic field  can be expressed as linear combinations of the background field through a symmetric magnetic polarizability tensor~\footnote{They justify the term magnetic polarizability tensor being applicable to this case as it is associated with a magnetic dipole and is a generalised susceptibility.}.  Subsequently,
J\"arvi~\cite{jarvi} and, independently, Baum~\cite{baum499}, have proposed that (\ref{eqn:engexp}) holds for general  shaped conducting objects~\footnote{Although the connection is not explicit, these authors appear to combine the proposition of Landau and Lifschitz~\cite{landau} and a dipole expansion of the field to recover this result.}. Moreover, to the best of our knowledge, with the exception of the sphere, an explicit expression for computing the tensor coefficients of an object is not available and this has led to the common engineering approach of empirically fitting their coefficients e.g.~\cite{das1990,norton2001,pasion2001,peyton2013}.

Ammari, Chen, Chen, Garnier and Volkov~\cite{ammarivolkov2013} have recently obtained an asymptotic expansion, which, for the first time, correctly
describes the perturbed magnetic field as $\alpha \to 0$ for a conducting (possibly also permeable and multiply connected) object in the
presence of a low--frequency background magnetic field, generated by a coil with an alternating current. We quote the result of their Theorem 3.2 in the alternative manner, as first stated in~\cite{ledgerlionheart2014}
%%%%%%%%%%%%%%%%%%%%%%%%%%%%%%%%%%%%%%%%%%%%%%%%%%%%%%%%%%%%%%%%%%%%%%%%%%%%%%%
\begin{align}
({\vec H}_\alpha - {\vec H}_0) ({\vec x} )_i =&  ({\vec D}^2 G( {\vec x} , {\vec z} ))_{\ell m}  {\mathcal M}_{ \ell m  i j} ({\vec H}_0 ( {\vec z} ))_j  
 +O(\alpha^4) , \label{eqn:ammarimain}
\end{align}
%%%%%%%%%%%%%%%%%%%%%%%%%%%%%%%%%%%%%%%%%%%%%%%%%%%%%%%%%%%%%%%%%%%%%%%%%%%%%%%
as $\alpha \to 0$. At first glance this appears similar to the result stated in (\ref{eqn:asympformula}), however, note that the result is expressed in terms of a rank 4 tensor, which appears as an inner product with ${\vec D}^2G({\vec x},{\vec z})$ over the first two indices. Moreover, a rank 4 tensor can have as many as 64 independent complex coefficients ${\mathcal M}_{ \ell m  j i}$ whereas for a symmetric rank 2 tensor this can have at most $6$. In~\cite{ledgerlionheart2014}, we show that for a right--handed orthonormal coordinate system, described by the unit vectors $\hat{\vec e}_j$, $j=1,2,3$, which are chosen as the Cartesian coordinate directions, then it is possible to reduce this result to 
%%%%%%%%%%%%%%%%%%%%%%%%%%%%%%%%%%%%%%%%%%%%%%%%%%%%%%%%%%%%%%%%%%%%%%%%%%%%%%%
\begin{align}
({\vec H}_\alpha - {\vec H}_0) ({\vec x} )_i =&  ({\vec D}^2 G( {\vec x} , {\vec z} ))_{ij }  \widecheck{\widecheck{\mathcal M}}_{ j k} ({\vec H}_0 ( {\vec z} ))_k 
                                                                     +O(\alpha^4)
 \label{eqn:mainresultb} ,
\end{align}
%%%%%%%%%%%%%%%%%%%%%%%%%%%%%%%%%%%%%%%%%%%%%%%%%%%%%%%%%%%%%%%%%%%%%%%%%%%%%%%
as $\alpha \to 0$
where $\widecheck{\widecheck{\mathcal M}}$ is a complex symmetric rank 2 tensor, which we henceforth denote as the \emph{magnetic polarizability tensor}~\footnote{A more precise definition would be to say that $\widecheck{\widecheck{\mathcal M}}$ is the leading order approximation to the magnetic polarizability tensor for small objects. In~\cite{ledgerlionheart2014} we include examples to demonstrate numerically that the prerturbed field, and hence the coefficients of  $\widecheck{\widecheck{\mathcal M}}$, behave asymptotically as predicated by (\ref{eqn:mainresultb}).}. For consistency with~\cite{ledgerlionheart2014} we use a single check to denote the reduction in a tensor's rank by 1 and a single hat to denote its extension in rank by 1. 
 The coefficients of  $\widecheck{\widecheck{\mathcal M}}$ are defined as  $\widecheck{\widecheck{\mathcal M}}_{ij} : =  {\mathcal N}_{ij} - \widecheck{\mathcal C}_{ij}$ where
\begin{subequations}
\begin{align}
\widecheck{\mathcal C}_{ij }  & : = - \frac{\im \nu \alpha^3 }{4}\hat{\vec e}_i \cdot \int_B {\vec \xi} \times ({\vec \theta}_j + \hat{\vec e}_j \times {\vec \xi} ) \dif {\vec \xi} , \\
{\mathcal N}_{ij } & := \alpha^3 \left ( 1- \frac{\mu_0}{\mu_*} \right ) \int_B \left (
\hat{\vec e}_i \cdot  \hat{\vec e}_j + \frac{1}{2}  \hat{\vec e}_i \cdot  \nabla \times {\vec \theta}_j \right ) \dif {\vec \xi} ,
\end{align} \label{eqn:defintiiontensor}
\end{subequations}
%%%%%%%%%%%%%%%%%%%%%%%%%%%%%%%%%%%%%%%%%%%%%%%%%%%%%%%%%%%%%%%%%%%%%%%%%%%%%%%
and ${\vec \theta}_i$ solves the vector valued transmission problem
%%%%%%%%%%%%%%%%%%%%%%%%%%%%%%%%%%%%%%%%%%%%%%%%%%%%%%%%%%%%%%%%%%%%%%%%%%%%%%%
\begin{subequations}
\begin{align} 
\nabla \times \mu^{-1} \nabla \times {\vec \theta}_i -  \im \omega \sigma \alpha^2 {\vec \theta}_i & = \im \omega \sigma \alpha^2 \hat{\vec e}_i \times {\vec \xi}   &&  \hbox{in $B \cup B^c$}  , \\
\nabla \cdot {\vec \theta}_i &= 0  && \hbox{in $B^c$} , \\
\left [ {\vec \theta} _i\times \hat{\vec n} \right ]_\Gamma   &= {\vec 0},  &&  \hbox{on $\Gamma$} ,\\
  \left [   \mu^{-1}   \nabla \times {\vec \theta}_i \times \hat{\vec n} \right ]_\Gamma  & = 
 -2 \left [ \mu^{-1} \right ]_\Gamma \hat{\vec e}_i \times \hat{\vec n}  &&
\hbox{on $\Gamma$} ,\\ 
{\vec \theta}_i( {\vec \xi}) & = O(|{\vec \xi} |^{-1}) && \hbox{as $|{\vec \xi}| \to \infty$} .
\end{align} \label{eqn:transproblemfull}
\end{subequations}
%%%%%%%%%%%%%%%%%%%%%%%%%%%%%%%%%%%%%%%%%%%%%%%%%%%%%%%%%%%%%%%%%%%%%%%%%%%%%%%
We emphasise that (\ref{eqn:ammarimain},\ref{eqn:mainresultb}) provides a rigorous mathematical framework  for the perturbed magnetic field for the eddy current case and  (\ref{eqn:defintiiontensor}) together with the solution of the transmission problem (\ref{eqn:transproblemfull}) now provide explicit expressions for the computation of the coefficients of the tensor $\widecheck{\widecheck{\mathcal M}}$. In~\cite{ledgerlionheart2014} we present numerical results for the computation of the tensor coefficients of different objects and describe how mirror and rotational symmetries of an object can be applied to reduce the number of independent coefficients of $\widecheck{\widecheck{\mathcal M}}$.

\subsection{Unified description for small objects} \label{sect:unimathmodel3}
From Sections~\ref{sect:mathmodel1} and~\ref{sect:mathmodel2} we observe that for small objects $ ({\vec H}_\alpha - {\vec H}_0)({\vec x} )$ has the form of (\ref{eqn:asympformula}) with ${\vec R}({\vec x}) = O(\alpha^4)$ and ${\mathcal A}= {\mathcal T} ( \mu_r )$ for the magnetostatic response and ${\mathcal A}= \widecheck{\widecheck{\mathcal M}}$ for the eddy current case as $\alpha \to 0$. 
We now state a series of lemmas, which unifies their treatment and shows that the former is, in fact, a simplification of the latter. We start with an alternative form of $\widecheck{\widecheck{\mathcal M}}$ :
%%%%%%%%%%%%%%%%%%%%%%%%%%%%%%%%%%%%%%%%%%%%%%%%%%%%%%%%%%%%%%%%%%%%%%%%%%%%%%%
\begin{lemma}
%%%%%%%%%%%%%%%%%%%%%%%%%%%%%%%%%%%%%%%%%%%%%%%%%%%%%%%%%%%%%%%%%%%%%%%%%%%%%%%
The coefficents of  $\widecheck{\widecheck{\mathcal M}}$ can be expressed as $\widecheck{\widecheck{\mathcal M}}_{ij} = 
 {\mathcal N}_{ij}^{\sigma_*}+ {\mathcal N}_{ij}^0 - \widecheck{{\mathcal C}}_{ij}^{\sigma_*}$ where
\begin{subequations}
\begin{align}
\widecheck{{\mathcal C}}_{ij}^{\sigma_*}  & : = - \frac{\im \nu \alpha^3 }{4}\hat{\vec e}_i \cdot \int_B {\vec \xi} \times ({\vec \theta}_j^{(0)} + {\vec \theta}_j^{(1)}  ) \dif {\vec \xi} , \\
{\mathcal N}_{ij}^{\sigma_*} & := \frac{\alpha^3}{2} \left ( 1- \frac{\mu_0}{\mu_*} \right ) \int_B \left (
    \hat{\vec e}_i \cdot  \nabla \times {\vec \theta}_j^{(1)} \right ) \dif {\vec \xi} \label{eqn:definencon}, \\
{\mathcal N}_{ij}^0 & := \frac{\alpha^3}{2} \left ( 1- \frac{\mu_0}{\mu_*} \right ) \int_B \left (
  \hat{\vec e}_i \cdot  \nabla \times {\vec \theta}_j^{(0)} \right ) \dif {\vec \xi} \label{eqn:definen0}.
\end{align}
\end{subequations}
%%%%%%%%%%%%%%%%%%%%%%%%%%%%%%%%%%%%%%%%%%%%%%%%%%%%%%%%%%%%%%%%%%%%%%%%%%%%%%%
and $ {\mathcal N}^{\sigma_*} -\widecheck{{\mathcal C}}^{\sigma_*}$ is a complex symmetric rank 2 tensor and ${\mathcal N}^0$ is a real symmetric rank 2 tensor. The coefficients of these tensors depend on the solutions ${\vec \theta}_i^{(0)}$, ${\vec \theta}_i^{(1)}$, $i=1,2,3$, to the transmission problems
%%%%%%%%%%%%%%%%%%%%%%%%%%%%%%%%%%%%%%%%%%%%%%%%%%%%%%%%%%%%%%%%%%%%%%%%%%%%%%%
\begin{subequations}
\begin{align} 
\nabla \times \mu^{-1} \nabla \times {\vec \theta}_i^{(0)} & ={\vec 0}  && \hbox{in $B \cup B^c$}  , \\
\nabla \cdot {\vec \theta}_i^{(0)} &= 0  && \hbox{in $B\cup B^c$} , \\
\left [ {\vec \theta} _i ^{(0)} \times \hat{\vec n} \right ]_\Gamma   &= {\vec 0}   && \hbox{on $\Gamma$} ,\\ 
\,  \left [   \mu^{-1}   \nabla  \times {\vec \theta}_i^{(0)}  \times \hat{\vec n} \right ]_\Gamma  & = 
{\vec 0}  &&
\hbox{on $\Gamma$} ,\\ 
{\vec \theta}_i^{(0)} ( {\vec \xi})- \hat{\vec e}_i \times {\vec \xi} & = O(|{\vec \xi} |^{-1})  && \hbox{as $|{\vec \xi}| \to \infty$} ,
\end{align} \label{eqn:transproblem0}
\end{subequations}
%%%%%%%%%%%%%%%%%%%%%%%%%%%%%%%%%%%%%%%%%%%%%%%%%%%%%%%%%%%%%%%%%%%%%%%%%%%%%%%
and
%%%%%%%%%%%%%%%%%%%%%%%%%%%%%%%%%%%%%%%%%%%%%%%%%%%%%%%%%%%%%%%%%%%%%%%%%%%%%%%
\begin{subequations}
\begin{align} 
\nabla \times \mu^{-1} \nabla \times {\vec \theta}_i ^{(1)} - \im \omega \sigma \alpha^2 ({\vec \theta}_i^{(1)} +{\vec \theta}_i^{(0)})  & = {\vec 0}   && \hbox{in $B \cup B^c$}  , \\
\nabla \cdot {\vec \theta}_i^{(1)} &= 0 && \hbox{in $B^c$} , \\
\left [ {\vec \theta} _i^{(1)} \times \hat{\vec n} \right ]_\Gamma  & = {\vec 0}  && \hbox{on $\Gamma$} ,\\ 
  \left [   \mu^{-1}   \nabla  \times {\vec \theta}_i^{(1)} \times \hat{\vec n} \right ]_\Gamma  & = 
 {\vec 0} &&
\hbox{on $\Gamma$} ,\\ 
{\vec \theta}_i^{(1)}( {\vec \xi}) & = O(|{\vec \xi} |^{-1}) && \hbox{as $|{\vec \xi}| \to \infty$} .
\end{align} \label{eqn:transproblem1}
\end{subequations}
%%%%%%%%%%%%%%%%%%%%%%%%%%%%%%%%%%%%%%%%%%%%%%%%%%%%%%%%%%%%%%%%%%%%%%%%%%%%%%%
%%%%%%%%%%%%%%%%%%%%%%%%%%%%%%%%%%%%%%%%%%%%%%%%%%%%%%%%%%%%%%%%%%%%%%%%%%%%%%%
\end{lemma}
%%%%%%%%%%%%%%%%%%%%%%%%%%%%%%%%%%%%%%%%%%%%%%%%%%%%%%%%%%%%%%%%%%%%%%%%%%%%%%%
\begin{proof}
%%%%%%%%%%%%%%%%%%%%%%%%%%%%%%%%%%%%%%%%%%%%%%%%%%%%%%%%%%%%%%%%%%%%%%%%%%%%%%%
The result immediately follows from the ansatz ${\vec \theta}_i = {\vec \theta}_i^{(0)} + {\vec \theta}_i^{(1)}-\hat{\vec e}_i \times {\vec \xi}$. The symmetries of  $ {\mathcal N}^{\sigma_*}-\widecheck{{\mathcal C}}^{\sigma_*}$ and  ${\mathcal N}^0$ follow similar arguments to the proof of Lemma 4.4 in~\cite{ledgerlionheart2014}.
%%%%%%%%%%%%%%%%%%%%%%%%%%%%%%%%%%%%%%%%%%%%%%%%%%%%%%%%%%%%%%%%%%%%%%%%%%%%%%%
\end{proof}
%%%%%%%%%%%%%%%%%%%%%%%%%%%%%%%%%%%%%%%%%%%%%%%%%%%%%%%%%%%%%%%%%%%%%%%%%%%%%%%

To understand the role played by the topology of an object and its complement we  recall the definition of Betti numbers (e.g.~\cite{bossavit1998,ledgerzaglmayr2010} and references therein). The zeroth Betti number $\beta_0(\Omega)$, is the number of connected parts of $\Omega$, which for a bounded connected region in ${\mathbb R}^3$ is always $1$. The first Betti number, $\beta_1(\Omega)$ is the genus, i.e. the number of handles and the second Betti number $\beta_2(\Omega)$ is one less than the connected 
parts of the boundary $\partial \Omega$, ie. the number of cavities. If we consider a situation where $\Omega $ has $\beta_1 ( \Omega)$ handles then a non-bounding orientated path, $\gamma_i(\Omega)$, known as a loop, can be associated with each handle $i$. For $\beta_1 (\Omega)$ loops, one can associate $\beta_1(\Omega)$ cuts of $\Omega$ that can be represented by Seifert surfaces as shown for the situation where $B$ is a solid torus in Figure~\ref{fig:torus}. If $ \Sigma_1 \cup \cdots \cup \Sigma_N)$ where $N:= \beta_1(\Omega)$ stands for the complete set of cuts we recall that a curl free field can only be represented by the gradient of a scalar field in $\Omega \setminus ( \Sigma_1 \cup \cdots \Sigma_N)$. Furthermore, we recall that a simply connected region has $\beta_0(\Omega)=1$, $\beta_1(\Omega)=\beta_2(\Omega)=0$ and a multiply connected region is one that is not simply connected. 

\begin{figure}
\begin{center}
$\begin{array}{c}
\includegraphics[width=3in]{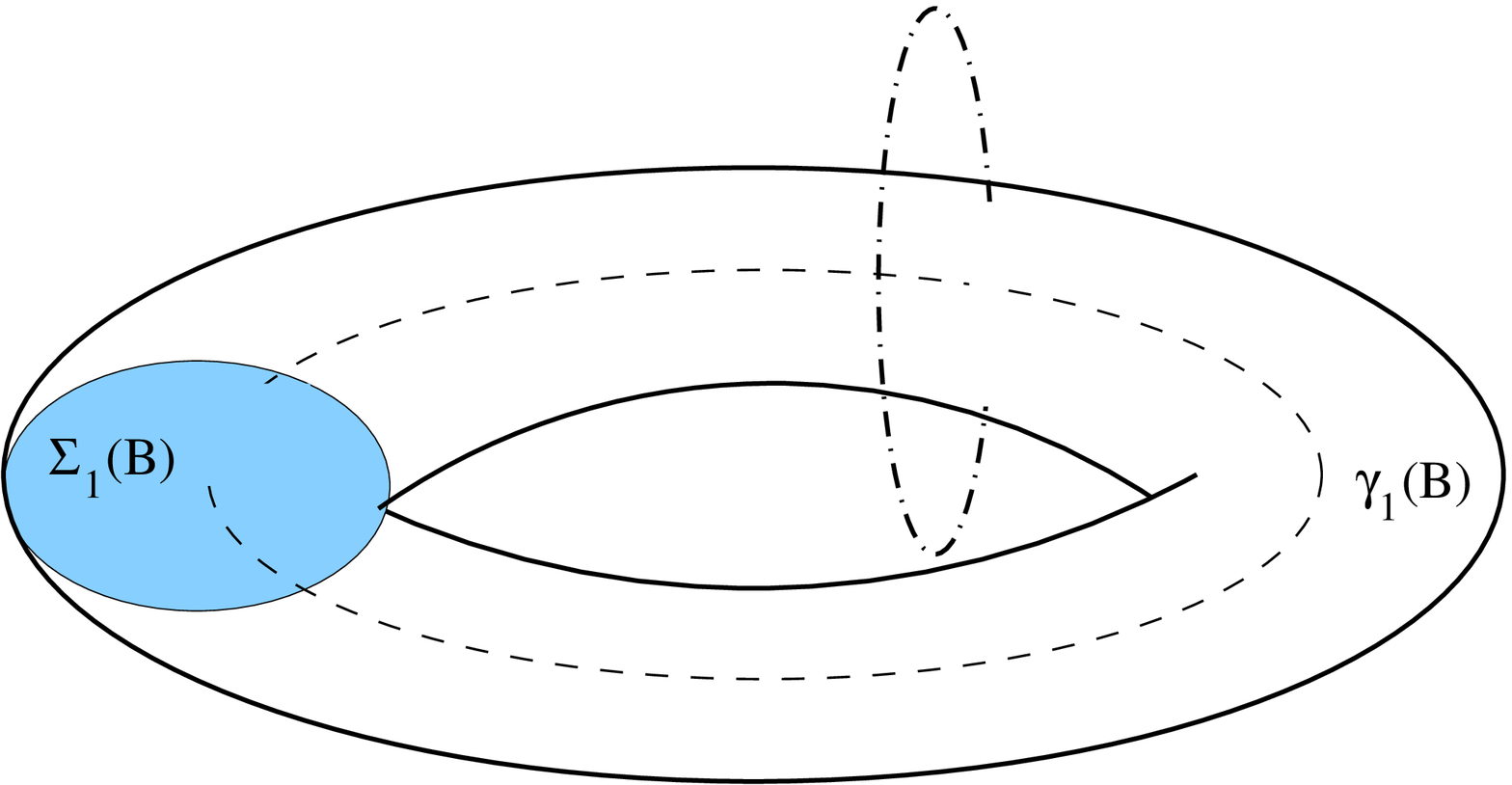} \\
(a)\\
\includegraphics[width=3in]{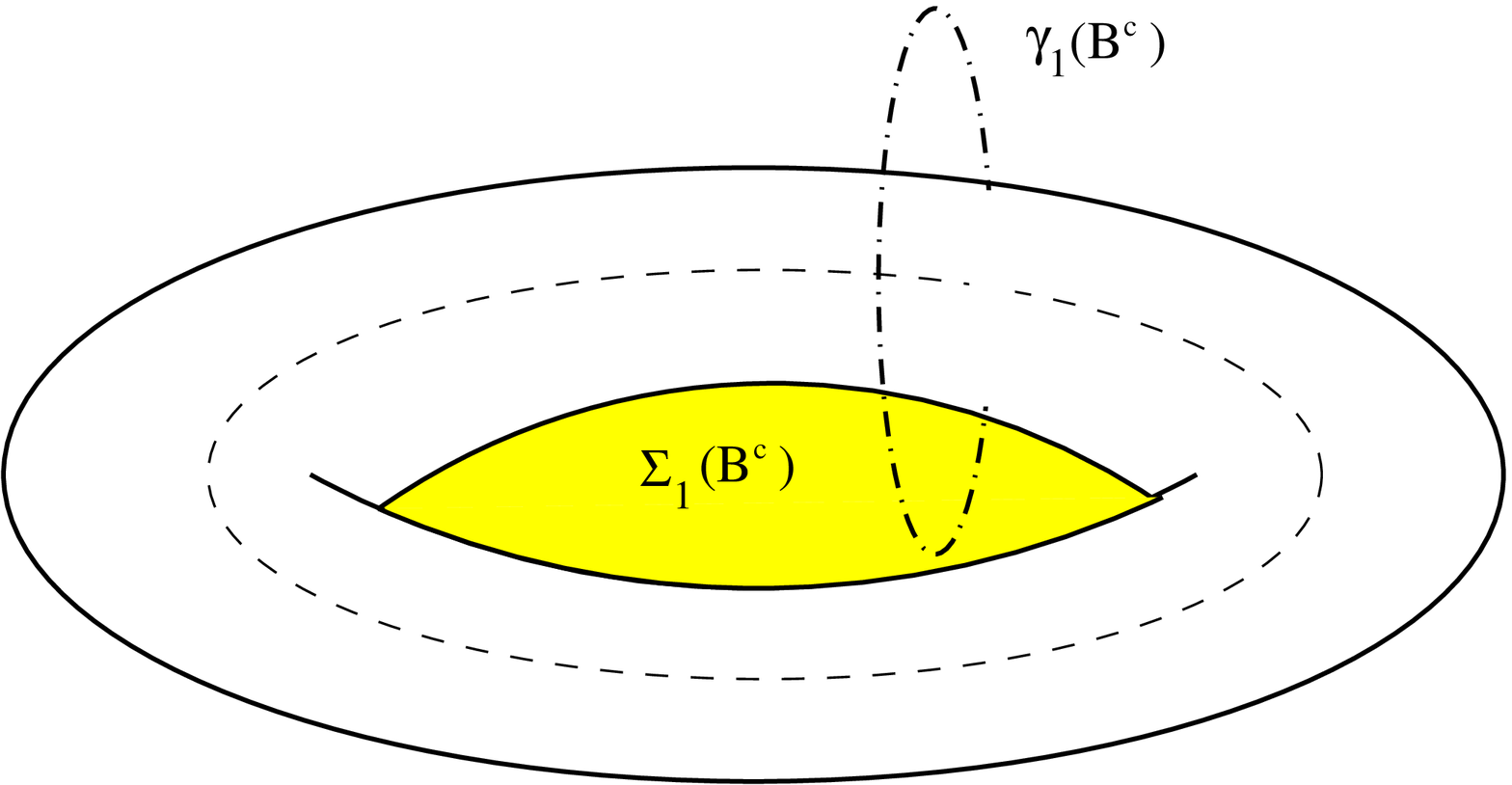} \\
(b)
\end{array}$
\end{center}
\caption{A torus shaped object $B$ showing $(a)$ a typical non--bounding cycle $\gamma_1(B)$ and associated cutting surface $\Sigma_1(B)$  for $B$ and $(b$) similar for $B^c$.} \label{fig:torus}
\end{figure}

\begin{remark}
In fact, in order to ensure uniqueness of (\ref{eqn:transproblemfull}) the additional condition
\begin{equation}
\int_{\Gamma_i} \left  . {\vec n} \cdot {\vec \theta}_i \right |^+ \dif {\vec \xi} = 0,
\end{equation}
where ${\Gamma_i}$, $i=1,\cdots,m$ and $m$ denotes the number of closed surfaces making up $\Gamma$, should be added. Analogous conditions should be added for the systems (\ref{eqn:transproblem0}), (\ref{eqn:transproblem1}).
\end{remark}

%%%%%%%%%%%%%%%%%%%%%%%%%%%%%%%%%%%%%%%%%%%%%%%%%%%%%%%%%%%%%%%%%%%%%%%%%%%%%%%
\begin{lemma} \label{lemma:reductiontosimp}
%%%%%%%%%%%%%%%%%%%%%%%%%%%%%%%%%%%%%%%%%%%%%%%%%%%%%%%%%%%%%%%%%%%%%%%%%%%%%%%
The tensor  $ {\mathcal N}^0$ reduces to ${\mathcal T}(\mu_r)$ independently of the geometric configuration of $B$, and, in particular, independently of the 
first Betti number of $B$ and $B^c$.
%%%%%%%%%%%%%%%%%%%%%%%%%%%%%%%%%%%%%%%%%%%%%%%%%%%%%%%%%%%%%%%%%%%%%%%%%%%%%%%
\end{lemma}
%%%%%%%%%%%%%%%%%%%%%%%%%%%%%%%%%%%%%%%%%%%%%%%%%%%%%%%%%%%%%%%%%%%%%%%%%%%%%%%
\begin{proof}
%%%%%%%%%%%%%%%%%%%%%%%%%%%%%%%%%%%%%%%%%%%%%%%%%%%%%%%%%%%%%%%%%%%%%%%%%%%%%%%
We introduce ${\vec v}_i := \nabla \times {\vec \theta}_i^{(0)} $ and set ${\vec u}_i = \tilde{\mu}_r^{-1} {\vec v}_i$ with  $\tilde{\mu}_r:= \left \{ \begin{array}{ll} \mu_r & \text{in $B$} \\ 1\ & \text{in $B^c$}  \end{array} \right .  $. In doing so we can establish, 
using the decay conditions of $\nabla \times {\vec \theta}_i=O(|{\vec \xi}|^{-3})$ as $|{\vec \xi}| \to \infty$~\cite{ammarivolkov2013}, that $ {\vec u}_i$ satisfies the transmission problem
%%%%%%%%%%%%%%%%%%%%%%%%%%%%%%%%%%%%%%%%%%%%%%%%%%%%%%%%%%%%%%%%%%%%%%%%%%%%%%%
 \begin{subequations} \label{eqn:firstordersys}
\begin{align}  
\nabla \times {\vec u}_i & = {\vec 0} &&\hbox{in $B\cup B^c$}, \\
\nabla \cdot (\tilde{\mu}_r {\vec u}_i ) & = 0 &&\hbox{in $B\cup B^c$} ,\\
[\hat{\vec n} \times {\vec u}_i ]_\Gamma &= {\vec 0}  && \hbox{on $\Gamma$} ,\\
 \ [ \hat{\vec n} \cdot (\tilde{\mu}_r {\vec u}_i  ) ]_\Gamma & =0 && \hbox{on $\Gamma$} ,\\
{\vec u}_i({\vec \xi}) - 2 \hat{\vec e}_i & = O (  | {\vec \xi}|^{-3}) && \hbox{as $|{\vec \xi}| \to \infty$} ,
\end{align}
\end{subequations} 
%%%%%%%%%%%%%%%%%%%%%%%%%%%%%%%%%%%%%%%%%%%%%%%%%%%%%%%%%%%%%%%%%%%%%%%%%%%%%%%
which is equivalent to (\ref{eqn:transproblem0}). We set the harmonic fields as ${\vec u}_i = \nabla \vartheta_i + {\vec h}_i$ in $B$ and $B^c$, where ${\vec h}_i$ represents curl free fields that are not gradients with $\text{dim}({\vec h}_i)=\beta_1(B)$ in $B$ and $\text{dim}({\vec h}_i)=\beta_1(B^c)$ in $B^c$. But, due to the fact that ${\vec u}_i$ is curl free for all of ${\mathbb R}^3$ in (\ref{eqn:firstordersys}), we have, {\em independent of the choice of loops $\gamma_i(B)$, $\gamma_j(B^c)$}, that
\begin{equation}
\oint_{\gamma_i(B)} \hat{\vec \tau} \cdot {\vec u}_i \dif {\vec \xi} = \oint_{\gamma_j(B^c)} \hat{\vec \tau} \cdot {\vec u}_i \dif {\vec \xi} = 0 , \nonumber
\end{equation}
 $i=1,\cdots, \beta_1(B)$, $j=1,\cdots,\beta_1(B^c)$, where $\hat{\vec \tau}$ denotes the unit tangent and thus, by Proposition 3 and Remark 3 in~\cite{bossavit1998}, ${\vec h}_i={\vec 0} $ in ${\mathbb R}^3$. Furthermore, by choosing $\vartheta_i:= 2(\mu_r-1)  \phi_i + 2 \xi_i$ then
 (\ref{eqn:firstordersys}) reduces to (\ref{eqn:transmisscal}) and
%%%%%%%%%%%%%%%%%%%%%%%%%%%%%%%%%%%%%%%%%%%%%%%%%%%%%%%%%%%%%%%%%%%%%%%%%%%%%%%
\begin{align}
{\mathcal N}_{ij}^0 & = {\alpha^3} \left ( \mu_r-1  \right ) \int_B \left (
  \hat{\vec e}_i \cdot 
  \left ( (\mu_r -1 ) \nabla \phi_j + \hat{\vec e}_j \right )
   \right ) \dif {\vec \xi} \nonumber \\
   & = \alpha^3\left ( (\mu_r -1 )  |B| \delta_{ij} + (\mu_r -1)^2 \int_B \hat{\vec e}_i \cdot \nabla \phi_j \dif {\vec \xi} \right ) \nonumber \\
   & = \alpha^3\left ( (\mu_r -1 )  |B|  \delta_{ij} + (\mu_r -1)^2 \int_\Gamma  \hat{\vec n}^-  \cdot \nabla \phi_j \xi_i \dif {\vec \xi} \right ) , \nonumber 
   \end{align}
%%%%%%%%%%%%%%%%%%%%%%%%%%%%%%%%%%%%%%%%%%%%%%%%%%%%%%%%%%%%%%%%%%%%%%%%%%%%%%%
where the last step follows by integration points. Finally, we get ${\mathcal N}^0 ={\mathcal T}(\mu_r)$ by the symmetry of the coefficients of the tensor~\cite{cedio1998}.
%%%%%%%%%%%%%%%%%%%%%%%%%%%%%%%%%%%%%%%%%%%%%%%%%%%%%%%%%%%%%%%%%%%%%%%%%%%%%%%
\end{proof}
%%%%%%%%%%%%%%%%%%%%%%%%%%%%%%%%%%%%%%%%%%%%%%%%%%%%%%%%%%%%%%%%%%%%%%%%%%%%%%%

%%%%%%%%%%%%%%%%%%%%%%%%%%%%%%%%%%%%%%%%%%%%%%%%%%%%%%%%%%%%%%%%%%%%%%%%%%%%%%%
\begin{remark}
%%%%%%%%%%%%%%%%%%%%%%%%%%%%%%%%%%%%%%%%%%%%%%%%%%%%%%%%%%%%%%%%%%%%%%%%%%%%%%%
By using the alternative splitting $\widecheck{\widecheck{\mathcal M}} =  {\mathcal N}^{\sigma_*}  - \widecheck{{\mathcal C}}^{\sigma_*}+{\mathcal N}^{0} $ and Lemma~\ref{lemma:reductiontosimp}
we can now write $\widecheck{\widecheck{\mathcal M}} ={\mathcal N}^{\sigma_*}  - \widecheck{{\mathcal C}}^{\sigma_*} + {\mathcal T}(\mu_r)$.
Whereas the original splitting would requires $\sigma_*=0$ for $\widecheck{\widecheck{\mathcal M}} = {\mathcal N} = {\mathcal T}( {\mu}_r)$. Thus the alternative splitting of $\widecheck{\widecheck{\mathcal M}}$ is useful as it allows us to separate the complex symmetric conducting part $ {\mathcal N}^{\sigma_*} -  \widecheck{{\mathcal C}}^{\sigma_*}  $  from the real symmetric magnetic part ${\mathcal N}^{ 0}={\mathcal T}(\mu_r)$  and associate the latter with the P\'olya--Szeg\"o tensor. 
We summarise the interrelationships between the different rank 2 tensors in Fig.~\ref{fig:famrank2ten} and emphasise that (\ref{eqn:mainresultb}) provides a unified description of $ ({\vec H}_\alpha - {\vec H}_0)({\vec x} )$ for eddy current and magnetostatic problems as the object size tends to zero for both simply connected and multiply connected objects.
%%%%%%%%%%%%%%%%%%%%%%%%%%%%%%%%%%%%%%%%%%%%%%%%%%%%%%%%%%%%%%%%%%%%%%%%%%%%%%%
\end{remark}
%%%%%%%%%%%%%%%%%%%%%%%%%%%%%%%%%%%%%%%%%%%%%%%%%%%%%%%%%%%%%%%%%%%%%%%%%%%%%%%

\begin{figure}[h]
\begin{center}
\includegraphics[width=4in]{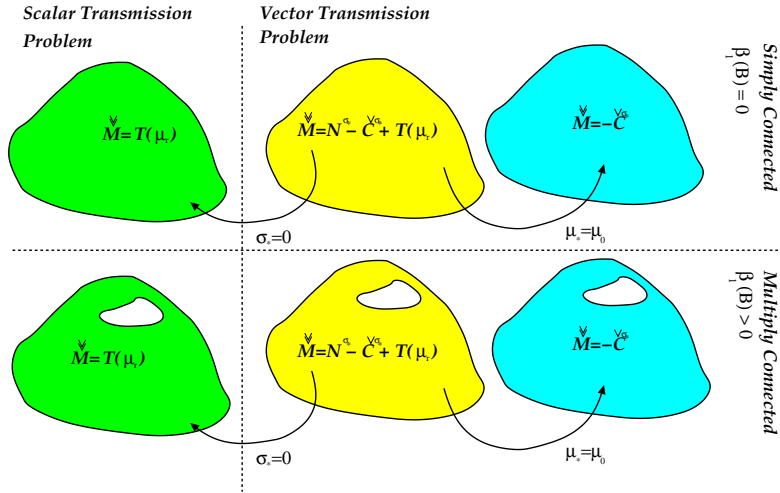}
\end{center}
\caption{A family of rank 2 polarization tensors for describing magnetic and conducting, simply and multiply connected, objects.} \label{fig:famrank2ten}
\end{figure}

\begin{corollary}
If the background magnetic field ${\vec H}_0 ( {\vec z} )$ is assumed to be that produced by a magnetic dipole, such as is appropriate for its evaluation at points away from a current source of small diameter centred at ${\vec y}$, then it follows that, at the centre of the object,
\begin{equation}
({\vec H}_0^e)_i:= ({\vec H}_0 ( {\vec z} ))_i =  ({\vec D}^2 G( {\vec y} , {\vec z} ) {\vec m}^e)_i ,  \nonumber
\end{equation}
where ${\vec m}^e$ is the magnetic dipole moment of exciting current source. Taking the component of (\ref{eqn:mainresultb}) in the direction ${\vec m}^m$ for this background field gives
\begin{align}
{\vec m}^m \cdot ({\vec H}_\alpha - {\vec H}_0)({\vec x}) &   ={\vec H}_0^m   \cdot ( \widecheck{\widecheck{\mathcal M}} {\vec H}_0^e)  + O(\alpha^4) ,
 \label{eqn:asymformeng}
\end{align}
as $\alpha \to 0$  where ${\vec H}_0^m := {\vec D}^2 G({\vec x},{\vec z}) {\vec m}^m$ is the background magnetic field, evaluated at the centre of the object, that would result from considering the measurement coil centred at ${\vec x}$ to be an excitor with dipole moment ${\vec m}^m$. We observe that the leading order term is exactly of the form quoted in (\ref{eqn:engexp}) and, moreover, by applying the Lorentz reciprocity theorem, it is possible to show that this is the leading term in the induced voltage
(see Appendix~\ref{sect:reciprocity}).
\end{corollary}

%%%%%%%%%%%%%%%%%%%%%%%%%%%%%%%%%%%%%%%%%%%%%%%%%%%%%%%%%%%%%%%%%%%%%%%%%%%%%%%
\section{Bounds on the tensor coefficients and associated properties} \label{sect:bounds}
%%%%%%%%%%%%%%%%%%%%%%%%%%%%%%%%%%%%%%%%%%%%%%%%%%%%%%%%%%%%%%%%%%%%%%%%%%%%%%%
\subsection{Preliminaries}
Following the restriction to orthonormal coordinates, we will, henceforth, arrange the coefficients of the rank 2 tensors ${\mathcal T}(\mu_r)$ and $\widecheck{\widecheck{\mathcal M}}$ (and its components ${\mathcal N}^0={\mathcal T}(\mu_r)$, ${\mathcal N}^{\sigma_*}$ and $\widecheck{\mathcal C}^{\sigma_*}$) in the form of $3 \times 3$ matrices. As standard, we shall compute their eigenvalues and eigenvectors of the tensors by computing the corresponding quantities for their matrix arrangements.

For some real contrast $k$, we recall that the P\'olya-Szeg\"o tensor ${\mathcal T}(k)$, is real symmetric with real eigenvalues and its eigenvectors are mutually perpendicular~\cite{ammarikangbook}.  It can be diagonalised as
%%%%%%%%%%%%%%%%%%%%%%%%%%%%%%%%%%%%%%%%%%%%%%%%%%%%%%%%%%%%%%%%%%%%%%%%%%%%%%%
\begin{subequations}
\begin{align}
({\mathcal R}^T  ({\mathcal T} {\mathcal R}))_{ij} = {\mathcal R}_{ki} {\mathcal T}_{k\ell} {\mathcal R}_{\ell j}  &= {\Lambda}_{ij}, \\
{\mathcal T}_{ij} = {\mathcal R}_{ik} {\Lambda}_{k\ell} {\mathcal R}_{j\ell}&= ({\mathcal R}(\Lambda {\mathcal R}^T))_{ij} \label{eqn:diagonal},
 \end{align}
 \end{subequations}
%%%%%%%%%%%%%%%%%%%%%%%%%%%%%%%%%%%%%%%%%%%%%%%%%%%%%%%%%%%%%%%%%%%%%%%%%%%%%%%
where ${\mathcal R}$ is real and orthogonal $({\mathcal R}^{-1}={\mathcal R}^T)$ and whose columns are the tensor's eigenvectors. Furthermore, ${\Lambda}$ is diagonal with its enteries $\lambda_1$, $\lambda_2$ and $\lambda_3$ being the eigenvalues of ${\mathcal T}$ (being positive definite for $1<k<\infty$ and negative definite for $0<k<1$). 

On the otherhand, as $\widecheck{\widecheck{\mathcal M}}$ is complex symmetric then, in general, it is not diagonalisable by a real rotation matrix apart for the specific case where the real and imaginary parts $\widecheck{\widecheck{\mathcal M}}$ commute such that $\widecheck{\widecheck{\mathcal M}} = \widecheck{\widecheck{\mathcal M}}_r + \im \widecheck{\widecheck{\mathcal M}}_r$ and in this case
\begin{subequations}
\begin{align}
({\mathcal R}^T  (\widecheck{\widecheck{\mathcal M}} {\mathcal R}))_{ij} = {\mathcal R}_{ki} (\widecheck{\widecheck{\mathcal M}}_r+ & \im \widecheck{\widecheck{\mathcal M}}_r)_{k\ell} {\mathcal R}_{\ell j}  = {\Lambda}_{ij}+ \im {\Lambda}_{ij}, \\
\widecheck{\widecheck{\mathcal M}}_{ij} = {\mathcal R}_{ik} ({\Lambda}_{k\ell} + \im {\Lambda}_{k\ell}   ) {\mathcal R}_{j\ell}&= ({\mathcal R}((\Lambda + \im \Lambda) {\mathcal R}^T))_{ij},
 \end{align}
 \end{subequations}
where the columns of ${\mathcal R}$ are the eigenvectors of  $\widecheck{\widecheck{\mathcal M}}_r$ and ${\Lambda}$ is diagonal with enteries $\lambda_1$, $\lambda_2$ and $\lambda_3$ being the eigenvalues of ${\mathcal M}_r$. The symmetric singular value decomposition~~\cite{thompson} can be applied to achieve diagonalisation of a complex symmetric matrix using a unitary matrix, although it remains to be shown whether this decomposition provides any practical insights.

We recall the Cayley--Hamilton theorem, which states that for a symmetric rank 2 tensor ${\mathcal A}$ that
\begin{equation}
{\mathcal A}^3 - I_1 {\mathcal A}^2 + I_2 {\mathcal A} - I_3 {\mathcal I} = 0,
\end{equation}
where it's invariants are $I_1 = \hbox{tr}\, ({\mathcal A})$,  $I_2= \frac{1}{2}( I_1^2 -  \hbox{tr}\, ( {\mathcal A}{\mathcal A}))$ and $I_3 = \hbox{det}\, {\mathcal A}$ and ${\mathcal I}_{ij}= \delta_{ij}$. 

%%%%%%%%%%%%%%%%%%%%%%%%%%%%%%%%%%%%%%%%%%%%%%%%%%%%%%%%%%%%%%%%%%%%%%%%%%%%%%%
\subsection{Some properties of the P\'olya--Szeg\"o tensor ${\mathcal T}(k)$} \label{sect:proppyola}
%%%%%%%%%%%%%%%%%%%%%%%%%%%%%%%%%%%%%%%%%%%%%%%%%%%%%%%%%%%%%%%%%%%%%%%%%%%%%%%

We first list some known properties of the P\'olya--Szeg\"o tensor  ${\mathcal T}(k)$ for $0< k \ne 1 < \infty$, which by Lemma~\ref{lemma:reductiontosimp}, also carry over to ${\mathcal N}^0$.
\begin{itemize}

\item Kleinman and Senior~\cite{kleinmansenior} show that the coefficients of the tensor satisfy
\begin{align}
\left ( {\mathcal T}_{ij} - \frac{k  -1 }{k +1} |B| \delta_{ij}\alpha^3 \right )^2 
\le  \left ( {\mathcal T}_{ii} -  \frac{k-1 }{k +1} |B| \alpha^3  \right ) & \left ( {\mathcal T}_{jj} -  \frac{k -1 }{k +1} |B| \alpha^3 \right ), \label{eqn:psoffdiag}
\end{align}
which is trivially satisfied for the diagonal entries. An alternative proof can be found in Ammari and Kang~\cite{ammarikangbook}.

\item Kleinman and Senior also show that the diagonal coefficients of the tensor satisfy 
\begin{equation}
\frac{k-1}{k } \le \frac{{\mathcal T}_{ii} }{\alpha^3 |B|} \le k -1, \label{eqn:diagelem}
\end{equation}
 Again, an alternative proof can be found in~\cite{ammarikangbook}, which also states that the eigenvalues of ${\mathcal T}$ satisfy the same inequality
\begin{equation}
\frac{k-1}{k} \le \frac{\lambda_i ( {\mathcal T})}{\alpha^3 |B|} \le k -1.  \label{eqn:evps}
\end{equation}

\item The bound on the trace of ${\mathcal T}$, which follows from (\ref{eqn:diagelem}), is not optimum and instead~\cite{ammarikangbook} proves the improved result, which we state for three dimensions below 
\begin{align}
\frac{1}{k -1} \hbox{tr}({\mathcal T}) & \le  \left ( 2+ \frac{1}{k} \right ) \alpha^3 |B|  \label{eqn:traceps},\\
(k-1) \hbox{tr}({\mathcal T}^{-1}) &\le \frac{2 +k}{|B| \alpha^3}  ,
\end{align}
and has been previously proved by Capdeboscq and Vogelius~\cite{capdeboscq2003}.
\end{itemize}

Using these results we establish the following.

%%%%%%%%%%%%%%%%%%%%%%%%%%%%%%%%%%%%%%%%%%%%%%%%%%%%%%%%%%%%%%%%%%%%%%%%%%%%%%%
\begin{lemma} \label{lemma:psprop}
%%%%%%%%%%%%%%%%%%%%%%%%%%%%%%%%%%%%%%%%%%%%%%%%%%%%%%%%%%%%%%%%%%%%%%%%%%%%%%%
For a contrast $1<k < \infty$ the invariants $I_1$, $I_2$ and $I_3$ of the P\'olya--Szeg\"o tensor ${\mathcal T}(k)$ satisfy
%%%%%%%%%%%%%%%%%%%%%%%%%%%%%%%%%%%%%%%%%%%%%%%%%%%%%%%%%%%%%%%%%%%%%%%%%%%%%%%
\begin{align}
0 &  <  I_1 \le \left ( 2+ \frac{1}{k} \right ) (k-1) \alpha^3 |B|,  \label{eqn:inv1kgr1}\\
0 &  <   |I_2 |  \le  \frac{1}{2} \left(    7+ \frac{4}{k} +\frac{1}{k^2} \right ) (k-1)^2  \alpha^6 |B|^2 , \label{eqn:inv2kgr1} \\
0  &< \alpha^9 |B|^3 \left ( 1- \frac{1}{k} \right )^3    \le I_3 \le (k-1)^3 \alpha^9 |B|^3 , \label{eqn:inv3kgr1}
\end{align}
%%%%%%%%%%%%%%%%%%%%%%%%%%%%%%%%%%%%%%%%%%%%%%%%%%%%%%%%%%%%%%%%%%%%%%%%%%%%%%%
in three dimensions. On the other hand, if $0 < k < 1 $ then the following inequalities hold
\begin{align}
0 &  >  I_1 \ge \left ( 2+ \frac{1}{k} \right ) (k-1) \alpha^3 |B| , \label{eqn:inv1klt1}\\
0 &  <  |I_2|  \le  \frac{1}{2} \left(  7+ \frac{4}{k} +\frac{1}{k^2} \right ) (k-1)^2  \alpha^6 |B|^2, \label{eqn:inv2klt1}  \\
 \alpha^9 & |B|^3 \left ( 1- \frac{1}{k} \right )^3    \le I_3 \le (k-1)^3 \alpha^9 |B|^3 <0. \label{eqn:inv3klt1}
 \end{align}
%%%%%%%%%%%%%%%%%%%%%%%%%%%%%%%%%%%%%%%%%%%%%%%%%%%%%%%%%%%%%%%%%%%%%%%%%%%%%%%
\end{lemma}
%%%%%%%%%%%%%%%%%%%%%%%%%%%%%%%%%%%%%%%%%%%%%%%%%%%%%%%%%%%%%%%%%%%%%%%%%%%%%%%

%%%%%%%%%%%%%%%%%%%%%%%%%%%%%%%%%%%%%%%%%%%%%%%%%%%%%%%%%%%%%%%%%%%%%%%%%%%%%%%
\begin{proof}
%%%%%%%%%%%%%%%%%%%%%%%%%%%%%%%%%%%%%%%%%%%%%%%%%%%%%%%%%%%%%%%%%%%%%%%%%%%%%%%
The results (\ref{eqn:inv1kgr1}) and (\ref{eqn:inv1klt1}) immediately follow from (\ref{eqn:traceps}). 

From (\ref{eqn:inv1kgr1}) $I_1^2= (\hbox{tr}({\mathcal T}))^2 \le \left ( 2+ \frac{1}{k} \right )^2 (k-1)^2 \alpha^6 |B|^2 $ for $0 < k \ne 1 < \infty$ and recalling (\ref{eqn:diagonal}) then $\hbox{tr}({\mathcal T}{\mathcal T} )= \hbox{tr}( {\mathcal R} \Lambda \Lambda {\mathcal R}^T) = \lambda_1^2 + \lambda_2^2+\lambda_3^2$. Thus,  since $|I_2| =\frac{1}{2}| I_1^2 -  (\hbox{tr}({\mathcal T}))^2| \le \frac{1}{2} ( I_1^2 + (\hbox{tr}({\mathcal T}))^2)$ and using (\ref{eqn:evps}), the results stated in  (\ref{eqn:inv2kgr1}) and (\ref{eqn:inv2klt1})   immediately follow. 

Recalling $I_3 = \hbox{det} ({\mathcal T}) = \lambda_1\lambda_2\lambda_3$ and that ${\mathcal T}$ is positive definite for $1 < k < \infty$ and negative definitive for $0 < k < 1$ then (\ref{eqn:inv3kgr1}) and (\ref{eqn:inv3klt1}) follow from immediately from the substitution of (\ref{eqn:evps}).
%%%%%%%%%%%%%%%%%%%%%%%%%%%%%%%%%%%%%%%%%%%%%%%%%%%%%%%%%%%%%%%%%%%%%%%%%%%%%%%
\end{proof}
%%%%%%%%%%%%%%%%%%%%%%%%%%%%%%%%%%%%%%%%%%%%%%%%%%%%%%%%%%%%%%%%%%%%%%%%%%%%%%%

%%%%%%%%%%%%%%%%%%%%%%%%%%%%%%%%%%%%%%%%%%%%%%%%%%%%%%%%%%%%%%%%%%%%%%%%%%%%%%%
\begin{corollary}
%%%%%%%%%%%%%%%%%%%%%%%%%%%%%%%%%%%%%%%%%%%%%%%%%%%%%%%%%%%%%%%%%%%%%%%%%%%%%%%
It immediately follows from Lemma~\ref{lemma:psprop} that the volumetric (spherical) part of ${\mathcal T}$ can be bounded as
%%%%%%%%%%%%%%%%%%%%%%%%%%%%%%%%%%%%%%%%%%%%%%%%%%%%%%%%%%%%%%%%%%%%%%%%%%%%%%%
\begin{equation}
3|B| ^2 \alpha^6 \left ( \frac{k-1}{k}\right )^2 \le \| \hbox{diag} ( {\mathcal T}) \|_F \le 3 |B|^2 \alpha^6 (k-1)^2,
%%%%%%%%%%%%%%%%%%%%%%%%%%%%%%%%%%%%%%%%%%%%%%%%%%%%%%%%%%%%%%%%%%%%%%%%%%%%%%%
\end{equation}
where $\| {\mathcal A} \|_F:=\displaystyle \sqrt{\sum_{i=1}^3\sum_{j=1}^3 |{\mathcal A}_{ij} |^2 }$ denotes the Forbenius matrix norm.
%%%%%%%%%%%%%%%%%%%%%%%%%%%%%%%%%%%%%%%%%%%%%%%%%%%%%%%%%%%%%%%%%%%%%%%%%%%%%%%
\end{corollary}
%%%%%%%%%%%%%%%%%%%%%%%%%%%%%%%%%%%%%%%%%%%%%%%%%%%%%%%%%%%%%%%%%%%%%%%%%%%%%%%

%%%%%%%%%%%%%%%%%%%%%%%%%%%%%%%%%%%%%%%%%%%%%%%%%%%%%%%%%%%%%%%%%%%%%%%%%%%%%%%
\begin{lemma}
%%%%%%%%%%%%%%%%%%%%%%%%%%%%%%%%%%%%%%%%%%%%%%%%%%%%%%%%%%%%%%%%%%%%%%%%%%%%%%%
The deviatoric part of ${\mathcal T}(k)$ can be bounded as
%%%%%%%%%%%%%%%%%%%%%%%%%%%%%%%%%%%%%%%%%%%%%%%%%%%%%%%%%%%%%%%%%%%%%%%%%%%%%%%
\begin{subequations}
\begin{align}
 \left \| {\mathcal T} - \frac{1}{3} \hbox{tr}({\mathcal T}) \right  \|_F^2 & \le \alpha^6 |B|^2 (k-1)^2 
  \left ( 3 +\frac{1}{3} \left ( 2 + \frac{1}{k} \right )^2 +6 \left ( 1 -  \frac{1}{k+1} \right )^2 \right ) ,\\
\hbox{if $1 < k < \infty$}& \hbox{ and as }\nonumber \\
\left \| {\mathcal T} - \frac{1}{3} \hbox{tr}({\mathcal T}) \right  \|_F^2 & \le \alpha^6 |B|^2 (k-1)^2 
  \left ( \frac{3}{k^2} +  \frac{1}{3} \left ( 2+ \frac{1}{k} \right )^2 + 6 \left ( \frac{1}{k} + \frac{1}{k+1}  \right )^2 \right ), \\
\hbox{if $ 0< k < 1$.} & {} \nonumber
\end{align}
\end{subequations}
%%%%%%%%%%%%%%%%%%%%%%%%%%%%%%%%%%%%%%%%%%%%%%%%%%%%%%%%%%%%%%%%%%%%%%%%%%%%%%%
\end{lemma}
%%%%%%%%%%%%%%%%%%%%%%%%%%%%%%%%%%%%%%%%%%%%%%%%%%%%%%%%%%%%%%%%%%%%%%%%%%%%%%%

%%%%%%%%%%%%%%%%%%%%%%%%%%%%%%%%%%%%%%%%%%%%%%%%%%%%%%%%%%%%%%%%%%%%%%%%%%%%%%%
\begin{proof}
We consider the case of $1 < k < \infty$, the proof for $0< k < 1$ is analogous. 
To show this we first fix $i=j$ then, by the triangular inequality and (\ref{eqn:diagelem}), (\ref{eqn:traceps}),
 it follows that
%%%%%%%%%%%%%%%%%%%%%%%%%%%%%%%%%%%%%%%%%%%%%%%%%%%%%%%%%%%%%%%%%%%%%%%%%%%%%%%
\begin{align}
\left ( {\mathcal T}_{ii} - \frac{1}{3} \hbox{tr}({\mathcal T} ) \right )^2  \le & {\mathcal T}_{ii}^2 + \frac{1}{9} \hbox{tr}({\mathcal T} ) ^2 
\le  \alpha^6 |B|^2 (k-1)^2 \left ( 1+ \frac{1}{9} \left (2 +\frac{1}{k} \right )^2  \right ).\label{eqn:fbnormdev1}
\end{align}
On the other hand, for $i\ne j$,  (\ref{eqn:psoffdiag}) implies that
\begin{align}
 {\mathcal T}_{ij}^2  = \left ( {\mathcal T}_{ij} - \frac{1}{3} \hbox{tr}({\mathcal T} )\delta_{ij}  \right )^2   \le & \left ( {\mathcal T}_{ii} - \alpha^3|B|\left (  \frac{k-1}{k+1} \right ) \right ) 
\left ( {\mathcal T}_{jj} - \alpha^3|B|\left (  \frac{k-1}{k+1} \right ) \right ) \nonumber \\
  \le &    \alpha^6|B|^2 (k-1)^2   \left ( 1- \frac{1}{k+1} \right )^2  \label{eqn:fbnormdev2},
\end{align}
by application of (\ref{eqn:diagelem}). Summing (\ref{eqn:fbnormdev1}) over the diagonal entries and (\ref{eqn:fbnormdev2}) over the off-diagonal entries completes the proof.

%%%%%%%%%%%%%%%%%%%%%%%%%%%%%%%%%%%%%%%%%%%%%%%%%%%%%%%%%%%%%%%%%%%%%%%%%%%%%%
\end{proof}
%%%%%%%%%%%%%%%%%%%%%%%%%%%%%%%%%%%%%%%%%%%%%%%%%%%%%%%%%%%%%%%%%%%%%%%%%%%%%%%

%%%%%%%%%%%%%%%%%%%%%%%%%%%%%%%%%%%%%%%%%%%%%%%%%%%%%%%%%%%%%%%%%%%%%%%%%%%%%%%
\section{Limiting low frequency and high conductivity response} \label{sect:lowhigh}
%%%%%%%%%%%%%%%%%%%%%%%%%%%%%%%%%%%%%%%%%%%%%%%%%%%%%%%%%%%%%%%%%%%%%%%%%%%%%%%
In this section we consider the low frequency and high conductivity limiting cases of $\widecheck{\widecheck{\mathcal M}}$. We also discuss the cases of high frequency and low conductivity. From the results presented in this section, we cannot necessarily deduce the behaviour of $({\vec H}_\alpha - {\vec H}_0) ({\vec x} )$  from (\ref{eqn:mainresultb}) since it is not permitted to substitute one asymptotic expansion, where $\alpha \to 0$ (and $\omega$ and $\sigma_*$ are fixed through $\nu$ ), in to another, where $\alpha $ is fixed and different limits on $\omega$ and $\sigma_*$ are taken. Still further, the eddy current model represents a quasi-static approximation to the Maxwell system and, in order that the modelling error is small, $\omega$ and $\sigma_*$ should be chosen according to shape dependent constants~\cite{schmidt2008}. Nonetheless, our theoretical results do have great practical relevence as the examples in Sections~\ref{sect:ellip}, \ref{sect:multfreq} and~\ref{sect:rem} will illustrate.
We first consider the low frequency response followed by the high conductivity case. 
%%%%%%%%%%%%%%%%%%%%%%%%%%%%%%%%%%%%%%%%%%%%%%%%%%%%%%%%%%%%%%%%%%%%%%%%%%%%%%%
\begin{theorem} \label{thm:lowfreq}
%%%%%%%%%%%%%%%%%%%%%%%%%%%%%%%%%%%%%%%%%%%%%%%%%%%%%%%%%%%%%%%%%%%%%%%%%%%%%%%
The low frequency limit for the coefficients of $\widecheck{\widecheck{\mathcal M}}$ can be described as 
%%%%%%%%%%%%%%%%%%%%%%%%%%%%%%%%%%%%%%%%%%%%%%%%%%%%%%%%%%%%%%%%%%%%%%%%%%%%%%%
\begin{equation}
\widecheck{\widecheck{\mathcal M}}_{ij} = {\mathcal N}_{ij}^0(\mu_r) +O(\omega)={\mathcal T}_{ij}(\mu_r) +O(\omega)
,
\end{equation}
%%%%%%%%%%%%%%%%%%%%%%%%%%%%%%%%%%%%%%%%%%%%%%%%%%%%%%%%%%%%%%%%%%%%%%%%%%%%%%%
as $\omega \to 0$ for an object $B$ with fixed conductivity $\sigma_*$ and relative permeability $\mu_r $. 
%%%%%%%%%%%%%%%%%%%%%%%%%%%%%%%%%%%%%%%%%%%%%%%%%%%%%%%%%%%%%%%%%%%%%%%%%%%%%%%
\end{theorem}
%%%%%%%%%%%%%%%%%%%%%%%%%%%%%%%%%%%%%%%%%%%%%%%%%%%%%%%%%%%%%%%%%%%%%%%%%%%%%%%

\begin{proof}
%%%%%%%%%%%%%%%%%%%%%%%%%%%%%%%%%%%%%%%%%%%%%%%%%%%%%%%%%%%%%%%%%%%%%%%%%%%%%%%
We recall the splittiing ${\vec \theta}_i = {\vec \theta}_i^{(0)} + {\vec \theta}_i^{(1)}-\hat{\vec e}_i \times {\vec \xi}$ and that $ {\vec \theta}_i^{(0)}$ and ${\vec \theta}_i^{(1)}$ solve 
(\ref{eqn:transproblem0}) and (\ref{eqn:transproblem1}), respectively. The weak form of the transmission problem for ${\vec \theta}_i^{(1)}$ is: Find ${\vec \theta}_i^{(1)} \in X$ such that
%%%%%%%%%%%%%%%%%%%%%%%%%%%%%%%%%%%%%%%%%%%%%%%%%%%%%%%%%%%%%%%%%%%%%%%%%%%%%%%
\begin{align}
\int_{\Omega}  \nabla \times {\vec \theta}_i^{(1)} \cdot & \nabla \times {\vec w}  -  \im \omega \sigma \mu \alpha^2 {\vec \theta}_i^{(1)} \cdot {\vec w} \dif {\vec \xi} = 
 \im \omega \sigma_* \mu_* \alpha^2 \int_B  {\vec \theta}_i^{(0)} \cdot {\vec w} \dif {\vec \xi} \qquad  \forall {\vec w} \in X, \label{eqn:weakform}
\end{align}
%%%%%%%%%%%%%%%%%%%%%%%%%%%%%%%%%%%%%%%%%%%%%%%%%%%%%%%%%%%%%%%%%%%%%%%%%%%%%%%
where $X:= \{ {\vec u} \in {\vec H}(\hbox{curl}(\Omega)) : \nabla \cdot {\vec u} = 0 \hbox{ in $\Omega$} \}$ and $\Omega=B^c \cup B$. Note that we have extended the requirement that 
${\vec \theta}_i^{(1)}$ be divergence free from $B^c$ to $\Omega$, but this is an immediate consequence of (\ref{eqn:transproblem1}a). Choosing ${\vec w}= 
({\vec \theta}_i^{(1)})^*$ in (\ref{eqn:weakform}), where $*$ denotes the complex conjugate, it follows for fixed $\alpha$, $\mu_*$, $\sigma_*$ that
%%%%%%%%%%%%%%%%%%%%%%%%%%%%%%%%%%%%%%%%%%%%%%%%%%%%%%%%%%%%%%%%%%%%%%%%%%%%%%%
\begin{align}
\| \nabla \times  {\vec \theta}_i^{(1)}\|_{L^2(\Omega)}^2   =  \int_{\Omega} |\nabla \times  {\vec \theta}_i^{(1)}|^2\dif {\vec \xi}  
 & \le 
\left | \int_{\Omega} |\nabla \times  {\vec \theta}_i^{(1)}|^2 - \im \omega \sigma \mu \alpha^2 | {\vec \theta}_i^{(1)}|^2 \dif {\vec \xi} \right | \nonumber \\
& =  \omega \sigma_* \mu_* \alpha^2  \left | \int_B  {\vec \theta}_i^{(0)} \cdot ( {\vec \theta}_i^{(1)})^* \dif {\vec \xi} \right | \nonumber \\
& \le  C \omega \| {\vec \theta}_i^{(0)} \|_{L^2(B)} \| {\vec \theta}_i^{(1)} \|_{L^2(B)}  \le  C \omega   \| {\vec \theta}_i^{(1)} \|_{L^2(\Omega)}, \nonumber
\end{align}
%%%%%%%%%%%%%%%%%%%%%%%%%%%%%%%%%%%%%%%%%%%%%%%%%%%%%%%%%%%%%%%%%%%%%%%%%%%%%%%
where $C$ is a generic constant independent of $\omega$ and ${\vec \theta}_i^{(1)}$ and the Cauchy-Schwartz inequality has been applied in the second to last step. Then, by 
using Corollary 3.51~\cite{monkbook} [pg72], and the far field decay of ${\vec \theta}_i^{(1)}$, we have $  \| {\vec \theta}_i^{(1)} \|_{L^2(\Omega)} \le C  \| \nabla \times {\vec \theta}_i^{(1)} \|_{L^2(\Omega)}$ so that 
%%%%%%%%%%%%%%%%%%%%%%%%%%%%%%%%%%%%%%%%%%%%%%%%%%%%%%%%%%%%%%%%%%%%%%%%%%%%%%%
\begin{equation}
\| \nabla \times  {\vec \theta}_i^{(1)}\|_{L^2(B)} \le  \| \nabla \times  {\vec \theta}_i^{(1)}\|_{L^2(\Omega)} \le C \omega.
\end{equation}
%%%%%%%%%%%%%%%%%%%%%%%%%%%%%%%%%%%%%%%%%%%%%%%%%%%%%%%%%%%%%%%%%%%%%%%%%%%%%%%
We use this result and the Cauchy-Schwartz inequality to establish the following
%%%%%%%%%%%%%%%%%%%%%%%%%%%%%%%%%%%%%%%%%%%%%%%%%%%%%%%%%%%%%%%%%%%%%%%%%%%%%%%
\begin{align}
| \widecheck{{\mathcal C}}_{ij}^{\sigma_*} | & = \frac{ \nu \alpha^3 }{4}  \left | \hat{\vec e}_i \cdot \int_B {\vec \xi} \times ({\vec \theta}_j^{(0)} + {\vec \theta}_j^{(1)}  ) \dif {\vec \xi} \right | 
\nonumber \\
& \le
 C\omega \left (  \left |  \int_B {\vec \theta}_j^{(0)} \cdot  {\vec \xi} \times\hat{\vec e}_i    \dif {\vec \xi} \right |  +
 \left |  \int_B  {\vec \theta}_j^{(1)} \cdot  {\vec \xi} \times\hat{\vec e}_i  \dif {\vec \xi} \right |  \right ) \nonumber\\
&  \le  C\omega \left (  \| {\vec \theta}_j^{(0)} \|_{L^2(B)}+ \| {\vec \theta}_j^{(1)} \|_{L^2(B)} \right ) 
  \le  C\omega \left (  1+ \| {\vec \theta}_j^{(1)} \|_{L^2(B)} \right ), \label{eqn:bdCsigma}
\end{align}
\begin{align}
| {\mathcal N}_{ij}^{\sigma_*}|  =\left |  \frac{\alpha^3}{2} \left ( 1- \frac{\mu_0}{\mu_*} \right ) \int_B \left (
    \hat{\vec e}_i \cdot  \nabla \times {\vec \theta}_j^{(1)} \right ) \dif {\vec \xi} \right |    
    \le C \| \nabla \times {\vec \theta}_j^{(1)} \|_{L^2(B)} \le C\omega , \label{eqn:bdNsigma}
    \end{align}
\begin{align}
| {\mathcal N}_{ij}^{0}|  =\left |  \frac{\alpha^3}{2} \left ( 1- \frac{\mu_0}{\mu_*} \right ) \int_B \left (
    \hat{\vec e}_i \cdot  \nabla \times {\vec \theta}_j^{(0)} \right ) \dif {\vec \xi} \right |     
     \le C \| \nabla \times {\vec \theta}_j^{(0)} \|_{L^2(B)} \le C \label{eqn:bdN0} .
    \end{align}
%%%%%%%%%%%%%%%%%%%%%%%%%%%%%%%%%%%%%%%%%%%%%%%%%%%%%%%%%%%%%%%%%%%%%%%%%%%%%%%
 Combining (\ref{eqn:bdCsigma}), (\ref{eqn:bdNsigma}) and (\ref{eqn:bdN0}) and using the decomposition $\widecheck{\widecheck{\mathcal M}}_{ij} = - \widecheck{{\mathcal C}}_{ij}^{\sigma_*} +
 {\mathcal N}_{ij}^{\sigma_*}+ {\mathcal N}_{ij}^0 $ the desired result immediately follows. The reduction to ${\mathcal T}_{ij}$ follows immediately from Lemma~\ref{lemma:reductiontosimp}.
%%%%%%%%%%%%%%%%%%%%%%%%%%%%%%%%%%%%%%%%%%%%%%%%%%%%%%%%%%%%%%%%%%%%%%%%%%%%%%%
 \end{proof}
%%%%%%%%%%%%%%%%%%%%%%%%%%%%%%%%%%%%%%%%%%%%%%%%%%%%%%%%%%%%%%%%%%%%%%%%%%%%%%%

\begin{remark} \label{remark:lowcond}
By following analogous steps one can also establish that $\widecheck{\widecheck{\mathcal M}}_{ij} = T(\mu_r) +O(\sigma_*) $ as $\sigma_* \to 0$ for an object $B$ with fixed relative permeability $\mu_r$ and frequency $\omega$. However, further to the comments at the beginning of this section, one needs to careful with the applicability of such a result.
\end{remark}

%%%%%%%%%%%%%%%%%%%%%%%%%%%%%%%%%%%%%%%%%%%%%%%%%%%%%%%%%%%%%%%%%%%%%%%%%%%%%%%
\begin{theorem} \label{thm:highfreq}
%%%%%%%%%%%%%%%%%%%%%%%%%%%%%%%%%%%%%%%%%%%%%%%%%%%%%%%%%%%%%%%%%%%%%%%%%%%%%%%
The high conductivity limit of the coefficients of $\widecheck{\widecheck{\mathcal M}}$ can be described as 
%%%%%%%%%%%%%%%%%%%%%%%%%%%%%%%%%%%%%%%%%%%%%%%%%%%%%%%%%%%%%%%%%%%%%%%%%%%%%%%
\begin{equation}
\widecheck{\widecheck{\mathcal M}}_{ij} = {\mathcal T}_{ij}(0) +O\left ( \frac{1}{\sqrt{\sigma_*}} \right ) \label{eqn:highcond} ,
\end{equation}
%%%%%%%%%%%%%%%%%%%%%%%%%%%%%%%%%%%%%%%%%%%%%%%%%%%%%%%%%%%%%%%%%%%%%%%%%%%%%%%
as $\sigma_* \to \infty$ for a object $B$, with $\beta_1(B)=\beta_1(B^c)=0$, fixed frequency $\omega$ and relative permeability $\mu_r $. Specifically,
\begin{equation}
 {\mathcal T}_{ij}(0)= \alpha^3 \left ( |B| \delta_{ij} - \int_\Gamma \hat{\vec n}^- \cdot \hat{\vec e}_i \psi_j \dif {\vec \xi} \right ) \label{eqn:pstensor0} ,
 \end{equation}
and $\psi_j$ solves 
 \begin{subequations}
 \begin{align}
 \nabla^2 \psi_j & = 0 && \hbox{in $B^c$} , \\
 \hat{\vec n} \cdot \nabla \psi_j & =   \hat{\vec n} \cdot \nabla \xi_j && \hbox{on $\Gamma$} ,\\
  \psi_j & \to 0  && \hbox{as $|{\vec \xi}| \to \infty$} .
 \end{align} \label{eqn:pstensor0tran}
 \end{subequations}
%%%%%%%%%%%%%%%%%%%%%%%%%%%%%%%%%%%%%%%%%%%%%%%%%%%%%%%%%%%%%%%%%%%%%%%%%%%%%%%
\end{theorem}
%%%%%%%%%%%%%%%%%%%%%%%%%%%%%%%%%%%%%%%%%%%%%%%%%%%%%%%%%%%%%%%%%%%%%%%%%%%%%%%

Before proving this result we first consider the following intermediate lemma.

\begin{lemma} \label{lemma:mcheckaltform}
The coefficients of $\widecheck{\widecheck{\mathcal M}}$ can be expressed as
%%%%%%%%%%%%%%%%%%%%%%%%%%%%%%%%%%%%%%%%%%%%%%%%%%%%%%%%%%%%%%%%%%%%%%%%%%%%%%%
\begin{equation}
\widecheck{\widecheck{\mathcal M}}_{ij} = \frac{\alpha^3}{2} \int_\Gamma \hat{\vec e}_i \cdot {\vec \theta}_j \times \hat{\vec n}^+ |_+ \dif {\vec \xi} - \frac{\alpha^3}{4} \int_\Gamma \hat{\vec e}_i \times {\vec \xi} \cdot \hat{\vec n}^+ \times \nabla \times {\vec \theta}_j |_+ \dif {\vec \xi} \label{eqn:mcheckaltform}.
\end{equation}
\end{lemma}

%%%%%%%%%%%%%%%%%%%%%%%%%%%%%%%%%%%%%%%%%%%%%%%%%%%%%%%%%%%%%%%%%%%%%%%%%%%%%%%
\begin{proof}
%%%%%%%%%%%%%%%%%%%%%%%%%%%%%%%%%%%%%%%%%%%%%%%%%%%%%%%%%%%%%%%%%%%%%%%%%%%%%%%
We begin by expressing $\widecheck{\mathcal C}_{ij }$ in an alternative form
%%%%%%%%%%%%%%%%%%%%%%%%%%%%%%%%%%%%%%%%%%%%%%%%%%%%%%%%%%%%%%%%%%%%%%%%%%%%%%%
\begin{align}
\widecheck{\mathcal C}_{ij }    =&  - \frac{\im \nu \alpha^3 }{4}\hat{\vec e}_i \cdot \int_B {\vec \xi} \times ({\vec \theta}_j + \hat{\vec e}_j \times {\vec \xi} ) \dif {\vec \xi} 
 =   - \frac{\alpha^3}{4 \mu_r} \int_B \nabla \times \nabla \times {\vec \theta}_j \cdot \hat{\vec e}_i \times {\vec \xi} \dif {\vec \xi} \nonumber \\
 = & - \frac{\alpha^3}{4} \int_\Gamma \hat{\vec e}_i \times {\vec \xi} \cdot \hat{\vec n}^- \times \mu_r^{-1} \nabla \times{\vec \theta}_j |_- \dif {\vec \xi}  
    + \frac{\alpha^3}{2\mu_r} \int_\Gamma \hat{\vec e}_i \cdot {\vec \theta}_j \times \hat{\vec n}^- |_- \dif {\vec \xi} \nonumber,
\end{align}
%%%%%%%%%%%%%%%%%%%%%%%%%%%%%%%%%%%%%%%%%%%%%%%%%%%%%%%%%%%%%%%%%%%%%%%%%%%%%%%
which follows from using ${\vec \theta}_j + \hat{\vec e}_j \times {\vec \xi} = \frac{1}{\im \nu \mu_r} \nabla \times \nabla \times {\vec \theta}_j$ in $B$ and performing integration by parts.  Then, by using the transmission conditions for ${\vec \theta}_j$,
%%%%%%%%%%%%%%%%%%%%%%%%%%%%%%%%%%%%%%%%%%%%%%%%%%%%%%%%%%%%%%%%%%%%%%%%%%%%%%%
\begin{align}
\widecheck{\mathcal C}_{ij }= & - \frac{\alpha^3}{4} \int_\Gamma \hat{\vec e}_i \times {\vec \xi} \cdot \left ( \hat{\vec n}^- \times \nabla \times{\vec \theta}_j |_+ 
 +2 ( 1-\mu_r^{-1}) \hat{\vec n}^- \times \hat{\vec e}_j \right ) \dif {\vec \xi}  
- \frac{\alpha^3}{2\mu_r} \int_\Gamma \hat{\vec e}_i \cdot {\vec \theta}_j \times \hat{\vec n}^+ |_+ \dif {\vec \xi}. \nonumber
\end{align}
%%%%%%%%%%%%%%%%%%%%%%%%%%%%%%%%%%%%%%%%%%%%%%%%%%%%%%%%%%%%%%%%%%%%%%%%%%%%%%%
On the other hand
%%%%%%%%%%%%%%%%%%%%%%%%%%%%%%%%%%%%%%%%%%%%%%%%%%%%%%%%%%%%%%%%%%%%%%%%%%%%%%%
\begin{align}
{\mathcal N}_{ij } & = \alpha^3 \left ( 1- \mu_r^{-1} \right ) \int_B \left (
\delta_{ij}  + \frac{1}{2}  \hat{\vec e}_i \cdot  \nabla \times {\vec \theta}_j \right ) \dif {\vec \xi} \nonumber \\
&=  \alpha^3 \left ( 1- \mu_r^{-1} \right )   \left (
|B| \delta_{ij}  + \frac{1}{2} \int  \hat{\vec e}_i \cdot   {\vec \theta}_j \times \hat{\vec n}^+ |_+\dif {\vec \xi}  \right ) , \nonumber
\end{align}
%%%%%%%%%%%%%%%%%%%%%%%%%%%%%%%%%%%%%%%%%%%%%%%%%%%%%%%%%%%%%%%%%%%%%%%%%%%%%%%
by integration by parts and application of a transmission condition. By realising that
\begin{equation}
\frac{\alpha^3}{2} (1 - \mu_r^{-1} )  \int_\Gamma \hat{\vec e}_i \times {\vec \xi} \cdot \hat{\vec n}^- \times \hat{\vec e}_j  \dif {\vec \xi} = - \alpha^3 |B| (1-\mu_r^{-1}) \delta_{ij}, \nonumber
\end{equation}
%%%%%%%%%%%%%%%%%%%%%%%%%%%%%%%%%%%%%%%%%%%%%%%%%%%%%%%%%%%%%%%%%%%%%%%%%%%%%%%
it follows
%%%%%%%%%%%%%%%%%%%%%%%%%%%%%%%%%%%%%%%%%%%%%%%%%%%%%%%%%%%%%%%%%%%%%%%%%%%%%%%
\begin{align}
\widecheck{\widecheck{\mathcal M}}_{ij} = & - \widecheck{\mathcal C}_{ij }+ {\mathcal N}_{ij } \nonumber \\
=&  \frac{\alpha^3}{2} \int_\Gamma \hat{\vec e}_i \cdot {\vec \theta}_j \times \hat{\vec n}^+ |_+ \dif {\vec \xi} 
+ \frac{\alpha^3}{4} \int_\Gamma \hat{\vec e}_i \times {\vec \xi} \cdot \hat{\vec n}^- \times \nabla \times {\vec \theta}_j |_+ \dif {\vec \xi} \nonumber,
\end{align}
%%%%%%%%%%%%%%%%%%%%%%%%%%%%%%%%%%%%%%%%%%%%%%%%%%%%%%%%%%%%%%%%%%%%%%%%%%%%%%%
from which immediately follows the desired result.
%%%%%%%%%%%%%%%%%%%%%%%%%%%%%%%%%%%%%%%%%%%%%%%%%%%%%%%%%%%%%%%%%%%%%%%%%%%%%%%
\end{proof}
%%%%%%%%%%%%%%%%%%%%%%%%%%%%%%%%%%%%%%%%%%%%%%%%%%%%%%%%%%%%%%%%%%%%%%%%%%%%%%%

\begin{proof}[Proof of Theorem~\ref{thm:highfreq}]
Consider the decomposition ${\vec \theta}_i ={\vec \Delta}_i+  \left \{ \begin{array}{cc} {\vec \chi}_i & \hbox{in $B^c$} \\
 {\vec \psi}_i & \hbox{in $B$} \end{array} \right .$ where
\begin{subequations}
\begin{align} 
\nabla \times \mu_0^{-1} \nabla \times {\vec \chi}_i & ={\vec 0}   && \hbox{in $B^c$}  , \\
\nabla \cdot {\vec \chi}_i &= 0   && \hbox{in $B^c$} , \\
 \nabla  \times {\vec \chi}_i  \times \hat{\vec n}  & = -2 \hat{\vec e}_i \times \hat{\vec n}
 &&
\hbox{on $\Gamma$} ,\\ 
{\vec \chi}_i ( {\vec \xi}) & = O(|{\vec \xi} |^{-1}) && \hbox{as $|{\vec \xi}| \to \infty$} ,
\end{align} \label{eqn:transproblem3}
\end{subequations}
%%%%%%%%%%%%%%%%%%%%%%%%%%%%%%%%%%%%%%%%%%%%%%%%%%%%%%%%%%%%%%%%%%%%%%%%%%%%%%%
and
%%%%%%%%%%%%%%%%%%%%%%%%%%%%%%%%%%%%%%%%%%%%%%%%%%%%%%%%%%%%%%%%%%%%%%%%%%%%%%%
\begin{subequations}
\begin{align} 
\nabla \times \mu_*^{-1} \nabla \times {\vec \psi}_i 
-  \im \omega \sigma_* \alpha^2 {\vec \psi}_i   & = \im \omega \sigma_* \alpha^2  {\vec e}_i \times {\vec \xi}    && \hbox{in $B $}  , \\
\nabla \cdot {\vec \psi}_i  &= 0  && \hbox{in $B$} , \\
\nabla  \times {\vec \psi}_i  \times \hat{\vec n}  & =  -2 \hat{\vec e}_i \times \hat{\vec n} +\frac{\mu_*}{\nu \mu_0}  \nabla \times {\vec \chi}_i \times \hat{\vec n}|_+  &&
\hbox{on $\Gamma$} ,
\end{align} \label{eqn:transproblem4}
\end{subequations}
%%%%%%%%%%%%%%%%%%%%%%%%%%%%%%%%%%%%%%%%%%%%%%%%%%%%%%%%%%%%%%%%%%%%%%%%%%%%%%%
such that the problems inside and outside the object completely decouple in the case of high conductivity since $\nu=\alpha^2\omega\sigma_*\mu_0$. The transmission problem for ${\vec \Delta}_i$ is
%%%%%%%%%%{\vec \chi}_j %%%%%%%%%%%%%%%%%%%%%%%%%%%%%%%%%%%%%%%%%%%%%%%%%%%%%%%%%%%%%%%%%%%%%
\begin{subequations}
\begin{align} 
\nabla \times \mu^{-1}  \nabla \times {\vec \Delta}_i - \im \omega \sigma \alpha^2 {\vec \Delta}_i  & = {\vec 0}    && \hbox{in $B \cup B^c$}  , \\
\nabla \cdot {\vec \Delta}_i &= 0  && \hbox{in $B \cup B^c$} , \\
[ {\vec \Delta}_i \times \hat{\vec n} ]_\Gamma &= {\vec 0} &&
\hbox{on $\Gamma$} ,\\
 [ \mu^{-1} \nabla  \times   {\vec \Delta}_i \times \hat{\vec n}]_\Gamma  &=\frac{1}{\nu \mu_0}   \nabla \times {\vec \chi}_i \times \hat{\vec n} &&
\hbox{on $\Gamma$} ,\\
{\vec \Delta}_i ( {\vec \xi}) & = O(|{\vec \xi} |^{-1}) && \hbox{as $|{\vec \xi}| \to \infty$} .
\end{align} \label{eqn:transproblem5}
\end{subequations}
%%%%%%%%%%%%%%%%%%%%%%%%%%%%%%%%%%%%%%%%%%%%%%%%%%%%%%%%%%%%%%%%%%%%%%%%%%%%%%%
 In a similar manner to the proof of Lemma~\ref{lemma:reductiontosimp}, we introduce ${\vec u}_i := \nabla \times {\vec \chi}_i$ but, in this case,
 we have only ${\vec u}_i =  2 (s-1)  \nabla \vartheta_i +{\vec h}_i$ in $B^c$ rather than ${\mathbb R}^3$, for some parameter $s$ where $\text{dim}({\vec h}_i)=\beta_1(B^c)$,
 and thus we can no longer establish that ${\vec h}_i={\vec 0}$ independent of the topology of $B$ and $B^c$. Therefore, we restrict ourselves to the situation of an object such that $\beta_1(B)=\beta_1(B^c)=0$ and, in this case, ${\vec u}_i =  2 (s-1)  \nabla \vartheta_i$  where $\vartheta_i$ is the solution to
 \begin{subequations}
 \begin{align}
 \nabla^2 \vartheta_i & = 0 && \hbox{in $B^c$} , \\
(s-1)  \hat{\vec n} \cdot \nabla \vartheta_i & =   - \hat{\vec n} \cdot \nabla \xi_i && \hbox{on $\Gamma$} , \\
  \vartheta_i & \to 0  && \hbox{as $|{\vec \xi}| \to \infty$} .
 \end{align}
 \end{subequations}

   %%%%%%%%%%%%%%%%%%%%%%%%%%%%%%%%%%%%%%%%%%%%%%%%%%%%%%%%%%%%%%%%%%%%%%%%%%%%%%%
We use the form of $\widecheck{\widecheck{\mathcal M}}_{ij}$ established in Lemma~\ref{lemma:mcheckaltform} and first write $ \widecheck{\widecheck{\mathcal M}}_{ij}=\widecheck{\widecheck{\mathcal M}}_{ij}^{\chi} +\widecheck{\widecheck{\mathcal M}}_{ij}^{\Delta} $ where
\begin{align}
\widecheck{\widecheck{\mathcal M}}_{ij}^{\chi} = & \frac{\alpha^3}{2} \int_\Gamma \hat{\vec e}_i \cdot {\vec \chi}_j   \times \hat{\vec n}^+ |_+ \dif {\vec \xi} 
- \frac{\alpha^3}{4} \int_\Gamma \hat{\vec e}_i \times {\vec \xi} \cdot \hat{\vec n}^+ \times \nabla \times {\vec \chi}_j |_+ \dif {\vec \xi} , \label{eqn:mchecktheta2}\\
\widecheck{\widecheck{\mathcal M}}_{ij}^\Delta = & \frac{\alpha^3}{2} \int_\Gamma \hat{\vec e}_i \cdot {\vec \Delta}_j \times \hat{\vec n}^+ |_+ \dif {\vec \xi} 
- \frac{\alpha^3}{4} \int_\Gamma \hat{\vec e}_i \times {\vec \xi} \cdot \hat{\vec n}^+ \times \nabla \times {\vec \Delta}_j |_+ \dif {\vec \xi} \label{eqn:mcheckdelta} .
\end{align}
Considering integration by parts on the first term in  (\ref{eqn:mchecktheta2})  we establish that 
%%%%%%%%%%%%%%%%%%%%%%%%%%%%%%%%%%%%%%%%%%%%%%%%%%%%%%%%%%%%%%%%%%%%%%%%%%%%%%%
\begin{align}
\frac{\alpha^3}{2} \int_\Gamma \hat{\vec e}_i \cdot {\vec \chi}_j  \times \hat{\vec n}^+ |_+ \dif {\vec \xi} & =    - \frac{\alpha^3}{2} \int_{B^c} \hat{\vec e}_i \cdot \nabla \times {\vec \chi}_j  \dif {\vec \xi} 
= - {\alpha^3} (s-1) \int_{B^c} \hat{\vec e}_i \cdot \nabla \vartheta_j  \dif {\vec \xi}   \nonumber \\
& =  - {\alpha^3} (s-1) \left (  \int_{B^c} \hat{\vec e}_i \cdot \nabla \theta_j  \dif {\vec \xi} +  \int_{B^c}  \vartheta_j  \nabla^2 \xi_i\dif {\vec \xi} \right ) \nonumber\\ 
& =  - {\alpha^3} (s-1)    \int_\Gamma \hat{\vec n}^+  \cdot \nabla\xi_i \vartheta_j \dif {\vec \xi}  \label{eqn:mchecksimp1} ,
\end{align}
%%%%%%%%%%%%%%%%%%%%%%%%%%%%%%%%%%%%%%%%%%%%%%%%%%%%%%%%%%%%%%%%%%%%%%%%%%%%%%%
and, for the second term, by substituting $\nabla \times {\vec \chi}_i = 2(s-1) \nabla \vartheta_i $ and using the condition $(s-1)  \hat{\vec n} \cdot \nabla \vartheta_i  =   - \hat{\vec n} \cdot \nabla \xi_i$ on $\Gamma$
%%%%%%%%%%%%%%%%%%%%%%%%%%%%%%%%%%%%%%%%%%%%%%%%%%%%%%%%%%%%%%%%%%%%%%%%%%%%%%%
\begin{align}
- \frac{\alpha^3}{4} \int_\Gamma \hat{\vec e}_i \times {\vec \xi} \cdot \hat{\vec n}^+ \times \nabla \times {\vec \chi}_j |_+ \dif {\vec \xi}
 &= - \frac{\alpha^3}{2} (s-1) \int_\Gamma \hat{\vec n}^+\cdot ( \nabla \vartheta_j  \times ( \hat{\vec e}_i \times {\vec \xi}) ) \dif {\vec \xi} \nonumber \\
 & =  {\alpha^3} (s-1)  \int_{B^c}   \nabla \vartheta_j   \cdot \nabla \xi_i   \dif {\vec \xi} 
=  {\alpha^3} (s-1)  \int_{\Gamma} \hat{\vec n}^+ \cdot  \nabla \vartheta_j  \xi_i \dif {\vec \xi} \nonumber \\
&= - {\alpha^3}   \int_{\Gamma} \hat{\vec n}^+ \cdot  \nabla \xi_j  \xi_i \dif {\vec \xi} = \alpha^3 |B| \delta_{ij} \nonumber  .
\label{eqn:mchecksimp2}
   \end{align}
%%%%%%%%%%%%%%%%%%%%%%%%%%%%%%%%%%%%%%%%%%%%%%%%%%%%%%%%%%%%%%%%%%%%%%%%%%%%%%%
Adding (\ref{eqn:mchecksimp1}) and (\ref{eqn:mchecksimp2}) for $s=0$ gives
%%%%%%%%%%%%%%%%%%%%%%%%%%%%%%%%%%%%%%%%%%%%%%%%%%%%%%%%%%%%%%%%%%%%%%%%%%%%%%%
\begin{align}
\widecheck{\widecheck{\mathcal M}}_{ij}^\chi = &   {\alpha^3}  |B| \delta_{ij} + {\alpha^3}    \int_\Gamma \hat{\vec n}^+  \cdot \nabla\xi_i \vartheta_j \dif {\vec \xi} \nonumber \\
= &
{\alpha^3}  |B| \delta_{ij} - {\alpha^3}     \int_\Gamma \hat{\vec n}^-  \cdot \nabla\xi_i \vartheta_j \dif {\vec \xi} ={\mathcal T}_{ij}(0) ,
\end{align}
%%%%%%%%%%%%%%%%%%%%%%%%%%%%%%%%%%%%%%%%%%%%%%%%%%%%%%%%%%%%%%%%%%%%%%%%%%%%%%%
where ${\mathcal T}_{ij}(0)$ is as defined in (\ref{eqn:pstensor0}) and the problem for $\vartheta_i$ becomes that for $\psi_i$ stated in (\ref{eqn:pstensor0tran}). To bound $\widecheck{\widecheck{\mathcal M}}_{ij}^\Delta$ we consider the weak form of the transmission problem for ${\vec \Delta}_i$: Find ${\vec \Delta}_i \in X$ such that
\begin{align}
\int_{\Omega} \tilde{\mu}_r^{-1} \nabla \times {\vec \Delta}_i \cdot \nabla \times {\vec w}  - \im \omega \sigma \mu_0 \alpha^2 {\vec \Delta}_i  \cdot {\vec w} \dif {\vec \xi} 
= -
\frac{1}{\nu} \int_\Gamma \hat{\vec n} \times \nabla \times {\vec \chi}_i \cdot {\vec w} \dif {\vec \xi} \qquad \forall {\vec w} \in X \label{eqn:weakform2} .
\end{align}
If we choose ${\vec w}= ({\vec \Delta}_i)^*$ then it follows for fixed $\alpha$, $\mu_*$, $\omega$ that
\begin{align}
\| \nabla \times {\vec \Delta}_i \|_{L^2(\Omega)}^2 \le & C \left | 
\int_{\Omega} \tilde{\mu}_r^{-1} | \nabla \times {\vec \Delta}_i |^2   - \im \omega \sigma \mu_0 \alpha^2 | {\vec \Delta}_i |^2 \dif {\vec \xi} \right | \nonumber \\
\le & \frac{C}{\sigma_*} \left |
\int_\Gamma \hat{\vec n}^+ \times \nabla \times {\vec \chi}_i \cdot ({\vec \Delta}_i)^* \dif {\vec \xi} \right | \nonumber \\
\le & \frac{C}{\sigma_*} \left | \int_{B^c} \nabla \times {\vec \chi}_i \cdot (\nabla \times {\vec \Delta}_i)^* \dif {\vec \xi} \right | \nonumber\\
\le & \frac{C}{\sigma_*} \| \nabla \times {\vec \chi}_i \|_{L^2(B^c)}  \| \nabla \times {\vec \Delta}_i\|_{L^2(B^c)} 
\le  \frac{C}{\sigma_*}   \| \nabla \times {\vec \Delta}_i \|_{L^2(\Omega)}  ,\label{eqn:inequcurldelta}
\end{align}
where integration by parts and then the Cauchy Schwartz in equality has been applied and $C$ is independent of $\sigma_*$. It then follows that
\begin{equation}
\| \nabla \times {\vec \Delta}_i \|_{L^2(B)} \le \frac{C}{\sigma_*} \qquad \hbox{and} \qquad \| \nabla \times {\vec \Delta}_i  \|_{L^2(B^c)} \le \frac{C}{\sigma_*} \nonumber .
\end{equation}
In a similar way we establish that $\|  {\vec \Delta}_i \|_{L^2(B)}^2 \le \frac{C}{\sigma_*}   \| \nabla \times {\vec \Delta}_i \|_{L^2(\Omega)} $ so that
\begin{equation}
\|  {\vec \Delta}_i \|_{L^2(B)} \le \frac{C}{\sigma_*^{3/2}}. 
\end{equation} 
We then write (\ref{eqn:mcheckdelta}) as
\begin{align}
\widecheck{\widecheck{\mathcal M}}_{ij}^{\Delta} =& -\frac{\alpha^3}{2}\int_{B^c} \hat{\vec e}_i \cdot \nabla \times {\vec \Delta}_j \dif {\vec \xi} 
 -\frac{\alpha^3}{4\nu \mu_r} \int_{\Gamma} \hat{\vec e}_i \times {\vec \xi} \cdot \hat{\vec n}^+ \times \nabla \times {\vec \chi}_j  |_+ \dif {\vec \xi}   \nonumber \\
 &  -\frac{\alpha^3}{4\mu_r} \int_\Gamma  \hat{\vec e}_i \times {\vec \xi} \cdot \hat{\vec n}^+ \times \nabla \times {\vec \Delta}_j |_-  \dif {\vec \xi}\nonumber \\
 =& -\frac{\alpha^3}{2}\int_{B^c} \hat{\vec e}_i \cdot \nabla \times {\vec \Delta}_j \dif {\vec \xi}+ \frac{\alpha^3}{2\nu\mu_r} \int_{B^c} \hat{\vec e}_i  \cdot \nabla \times {\vec \chi}_j  \dif {\vec \xi}   \nonumber \\
 &  + \frac{\alpha^3}{4}  \int_B  \left ( \im\nu {\vec \Delta}_j \cdot \hat{\vec e}_i \times {\vec \xi} + 2\nabla \times {\vec \Delta}_j \cdot \hat{\vec e}_i \right )  \dif {\vec \xi} \nonumber .
  \end{align}
  Thus, by the Cauchy Schwartz inequality, it follows that 
\begin{align}
|\widecheck{\widecheck{\mathcal M}}_{ij}^{\chi}| \le & C +C \| \nabla \psi_j \|_{L^2(B^c)} \le C \ \nonumber \\ 
|\widecheck{\widecheck{\mathcal M}}_{ij}^{\Delta}| \le & C \|   \nabla \times {\vec \Delta}_j \|_{L^2(B^c)} + \frac{C}{\sigma_*} \|    \nabla \times {\vec \chi}_j  \|_{L^2(B^c)} + C \sigma_* \|     {\vec \Delta}_j \|_{L^2(B)} +C\|   \nabla \times {\vec \Delta}_j \|_{L^2(B)} \nonumber \\
\le &  \frac{C}{\sigma_*}  +C \sigma_* \|  {\vec \Delta}_j \|_{L^2(B )} 
\le     \frac{C}{\sigma_* }  +\frac{C}{\sqrt{\sigma_*}} \le \frac{C}{\sqrt{\sigma_*}} \nonumber ,
\end{align}
where $C$ does not depend on $\sigma_*$ and consequently $\widecheck{\widecheck{\mathcal M}}_{ij}= \widecheck{\widecheck{\mathcal M}}_{ij}^{\chi} + \widecheck{\widecheck{\mathcal M}}_{ij}^\Delta = \widecheck{\widecheck{\mathcal M}}_{ij}^{\chi} + O(1/\sqrt{\sigma_*}) $ as $\sigma_* \to \infty$ and consequently the result stated in (\ref{eqn:highcond}) directly follows.
\end{proof}
We summarise our results for the limiting cases of low frequencies and high conductivities in Figure~\ref{fig:famrank2tenfreq}.

\begin{remark} \label{remark:highfreq}
We could apply similar arguments and establish that $\widecheck{\widecheck{\mathcal M}}_{ij}= {\mathcal T}(0)_{ij} + O(1/\sqrt{\omega}) $ as $\omega \to \infty$ for an object with $\beta_1(B)=\beta_1(B^c)=0$ and fixed relative permeability $\mu_r$ and conductivity $\sigma_*$. However, similar to Remark~\ref{remark:lowcond} we need to be careful with the applicability of such a result. We will return to this point in Section~\ref{sect:ellip}.
\end{remark}

\begin{remark}
Noting that the coefficients of the tensors ${\mathcal T}(\mu_r)_{ij} $ and $ {\mathcal T}(0)_{ij} $ are real valued then, the behaviour obtained at low frequencies, 
$\widecheck{\widecheck{\mathcal M}}_{ij}= {\mathcal T}(\mu_r)_{ij} + O(\omega)$ as $\omega \to 0$, and at high frequencies, $\widecheck{\widecheck{\mathcal M}}_{ij}= {\mathcal T}(0)_{ij} + O( 1/ \sqrt{\omega})$ as $\omega \to \infty$, ties in with explanation in Landau and Lipshitz~\cite[p. 192]{landau}, who predict that the imaginary component of the magnetic polarizability tensor is proportional to $\omega$ as $\omega \to 0$ and is proportional $1/\sqrt{\omega}$ as $\omega \to \infty$ for a simply connected object. Furthermore, they argue that, as $\omega \to \infty$, the magnetic polarizability tensor becomes that of a super conductor. The coefficients $ {\mathcal T}(0)_{ij} $ are those of the P\'oyla-Szeg\"o tensor and are already known to be associated with the magnetic response of a simply connected perfect conductor~\cite{ledgerlionheart2012}. A super conductor being the case of a perfect conductor with zero magnetic field in the bulk of the conductor.
\end{remark}

\begin{figure}[h]
\begin{center}
\includegraphics[width=4in]{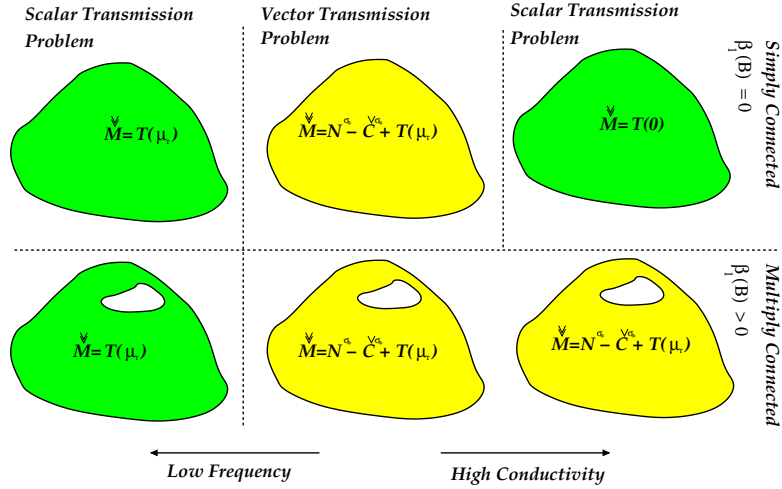}
\end{center}
\caption{Limiting cases of low frequency and high conductivity for simply and multiply connected objects.} \label{fig:famrank2tenfreq}
\end{figure}

%%%%%%%%%%%%%%%%%%%%%%%%%%%%%%%%%%%%%%%%%%%%%%%%%%%%%%%%%%%%%%%%%%%%%%%%%%%%%%%
\section{Elliptical objects}\label{sect:ellip}
%%%%%%%%%%%%%%%%%%%%%%%%%%%%%%%%%%%%%%%%%%%%%%%%%%%%%%%%%%%%%%%%%%%%%%%%%%%%%%%

For an ellipsoidal object $B_\alpha$ defined by
%%%%%%%%%%%%%%%%%%%%%%%%%%%%%%%%%%%%%%%%%%%%%%%%%%%%%%%%%%%%%%%%%%%%%%%%%%%%%%%
\begin{equation}
\frac{x_1^2}{a^2} + \frac{x_2^2}{b^2} + \frac{x_3^3}{c^2} = 1, \qquad 0 < c \le b \le a,
\nonumber
\end{equation}
%%%%%%%%%%%%%%%%%%%%%%%%%%%%%%%%%%%%%%%%%%%%%%%%%%%%%%%%%%%%%%%%%%%%%%%%%%%%%%%
and whose principal axes are chosen to coincide with the Cartesian coordinates axes, then an analytical solution is known for the P\'olya Szeg\"o tensor~\cite{ammarikangbook}
%%%%%%%%%%%%%%%%%%%%%%%%%%%%%%%%%%%%%%%%%%%%%%%%%%%%%%%%%%%%%%%%%%%%%%%%%%%%%%%
\begin{equation}
{\mathcal T} (k ) = \frac{4\pi a b c}{3}     \left (  \begin{array}{ccc}
\frac{(k -1 )}{ (1-A_1)+k A_1} & 0 & 0 \\
0 & \frac{(k -1 )}{ (1-A_2)+k A_2}  & 0 \\
0 & 0 & \frac{(k -1 )}{(1-A_3)+k A_3}   \end{array} \right ) \nonumber,
\end{equation}
%%%%%%%%%%%%%%%%%%%%%%%%%%%%%%%%%%%%%%%%%%%%%%%%%%%%%%%%%%%%%%%%%%%%%%%%%%%%%%%
where in the above the constants $A_1,A_2$ and $A_3$ are defined by
%%%%%%%%%%%%%%%%%%%%%%%%%%%%%%%%%%%%%%%%%%%%%%%%%%%%%%%%%%%%%%%%%%%%%%%%%%%%%%%
\begin{align*}
A_1& = \frac{bc}{a^2}\int_1^{+\infty} \frac{1}{ t^2 \sqrt{ t^2 -1 +\left ( \frac{b}{a} \right )^2} \sqrt{ t^2 -1 + \left ( \frac{c}{a} \right )^2 }} \dif t , \\
A_2& = \frac{bc}{a^2}\int_1^{+\infty} \frac{1}{ t^2 \left ( t^2 -1 +\left ( \frac{b}{a} \right )^2\right )^{3/2}  \sqrt{ t^2 -1 + \left ( \frac{c}{a} \right )^2 }} \dif t ,\\
A_3& = \frac{bc}{a^2}\int_1^{+\infty} \frac{1}{ t^2 \sqrt{ t^2 -1 +\left ( \frac{b}{a} \right )^2}  \left ( t^2 -1 + \left ( \frac{c}{a} \right )^2 \right )^{3/2}} \dif t ,
\end{align*}
%%%%%%%%%%%%%%%%%%%%%%%%%%%%%%%%%%%%%%%%%%%%%%%%%%%%%%%%%%%%%%%%%%%%%%%%%%%%%%%
or, by manipulation of the integrals, $A_1, A_2$ and $A_3$ can be expressed as 
%%%%%%%%%%%%%%%%%%%%%%%%%%%%%%%%%%%%%%%%%%%%%%%%%%%%%%%%%%%%%%%%%%%%%%%%%%%%%%%
\begin{equation}
A_1= \frac{abc}{2} d_1, \qquad A_2= \frac{abc}{2} d_2, \qquad A_3= \frac{abc}{2} d_3,
\end{equation}
where $d_j$, $j=1,2,3$ are the depolarizaing/demagnetising factors as defined in~\cite{milton}(pg 128). Alternatively $A_j$, $j=1,2,3$ can also be expressed in terms of elliptic integrals e.g.~\cite{osborn}. By considering the case where $k={\mu}_r \to 0$ the P\'olya Szeg\"o tensor ${\mathcal T}(0)$ describes the (magnetic) response for a perfectly conducting object
%%%%%%%%%%%%%%%%%%%%%%%%%%%%%%%%%%%%%%%%%%%%%%%%%%%%%%%%%%%%%%%%%%%%%%%%%%%%%%%
\begin{equation}
 {\mathcal T} (0 ) =  - \frac{4}{3} \pi a b c \left (  \begin{array}{ccc}
\frac{1}{  1-A_1} & 0 & 0 \\
0 & \frac{1}{1-  A_2}  & 0 \\
0 & 0 & \frac{1}{1- A_3}   \end{array} \right ) \nonumber .
\end{equation}
%%%%%%%%%%%%%%%%%%%%%%%%%%%%%%%%%%%%%%%%%%%%%%%%%%%%%%%%%%%%%%%%%%%%%%%%%%%%%%%
Furthermore, on consideration of the diagonalisation property of ${\mathcal T}( {\mu}_r)$ in (\ref{eqn:diagonal}), and the well known result that $({\vec x})_i \Lambda_{ij} ({\vec x})_j=1$ defines an ellipsoid aligned with the coordinate axes with semi principal axes lengths $1/\sqrt{\lambda_1}$, $1/\sqrt{\lambda_2}$ and $1/\sqrt{\lambda_3}$,  Khairuddin and Lionheart~\cite{taufiq2013} propose a strategy for determining an equivalent ellipsoid, which has the {\em same} P\'olya-Szeg\"o polarization tensor as the object under consideration.

Less is known for the case of the conducting (permeable) ellipsoid in the eddy current regime. An analytical solution for conducting (permeable)  prolate (and oblate) spheroids is available~\cite{aospheroid}, but, for numerical calculation, they require truncation of an otherwise infinitely sized linear system. An approximate solution approach for the conducting permeable spheroids~\cite{barrowes2004} and  ellipsoids~\cite{barrowes2008} has been proposed, but, is limited to the case where the objects have small skin depths. To the best of the authors' knowledge an analytical solution for the general ellipsoid  is not available. Thus an extension of the approach of  Khairuddin and Lionheart  to conducting ellipsoids is not immediate. 

In Fig.~\ref{fig:spheroidnonmag}, we show an illustration of the dependence of the diagonal coefficients of $ \widecheck{\widecheck{\mathcal M}} $ with frequency  for a prolate spheroid defined by $a=0.02\hbox{m}$, $b=c=0.01\hbox{m}$, $\sigma_*=1.5 \times10^7 \hbox{S/m}$ and $\mu_*=\mu_0$. We include comparisons between  the analytical solution of~\cite{aospheroid}~\footnote{The analytical solution for the magnetic polarizability tensor in this case is not a closed form expression but, instead requires computational truncation of an otherwise infinite system}, the converged numerical solution obtained by performing a frequency sweep using the approach~\cite{ledgerlionheart2014}, based on an unstructured grid of $14\,579$ tetrahedra and uniform $p=3$ elements, and ${\mathcal T} (0 )$ as the limit of the quasi-static approximation. 
  For the purpose of numerical calculation we truncate the computational domain at $100 |B|$ and employ the NETGEN mesh generator for this and other meshes used generated in this work~\cite{netgen}.
Following Remark~\ref{remark:highfreq}, we compute the shape dependent constants according to~\cite{schmidt2008}, and establish that quasi--static model remains valid for this spheroid provided that $f =\omega / ( 2 \pi) \ll 5.4 \text{Mhz}$.

The real part of the diagonal coefficients of  $ \widecheck{\widecheck{\mathcal M}} $ resemble a sigmoid function whereas the imaginary part of the coefficients have a peak value around $2\hbox{kHz}$ and vanish for low and high frequencies. The coefficients for the real part of $\widecheck{\widecheck{\mathcal M}} $ tend to zero at low frequencies and tend to those of $ {\mathcal T} (0 )$ at high frequencies. Thus agreeing with the theoretical predictions in Section~\ref{sect:lowhigh}. The agreement between the numerical prediction and the analytical solution is excellent.
%%%%%%%%%%%%%%%%%%%%%%%%%%%%%%%%%%%%%%%%%%%%%%%%%%%%%%%%%%%%%%%%%%%%%%%%%%%%%%%
\begin{figure*}
\begin{center}
$\begin{array}{cc}
\includegraphics[width=3in]{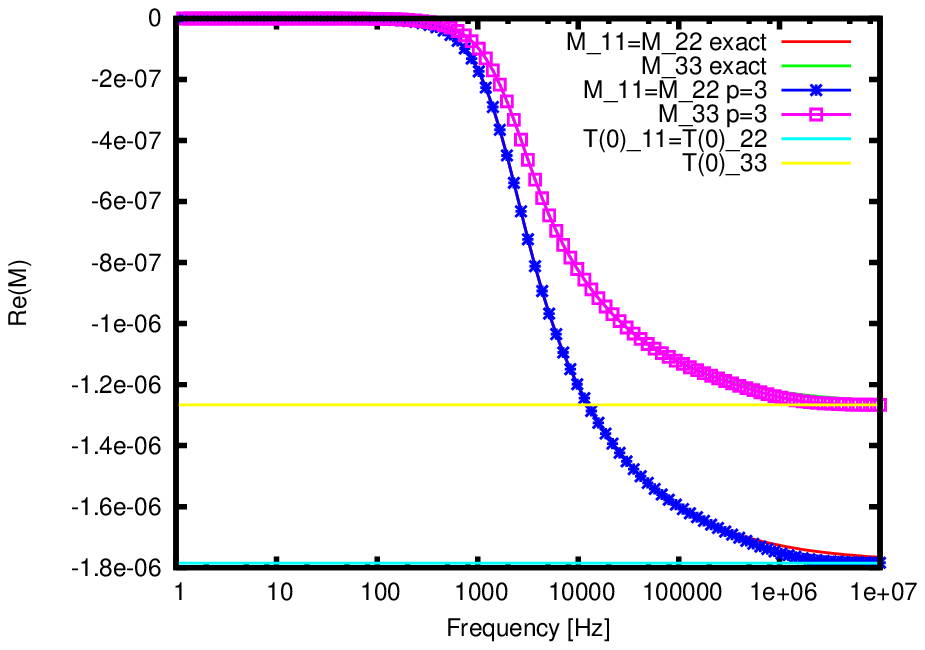} &
\includegraphics[width=3in]{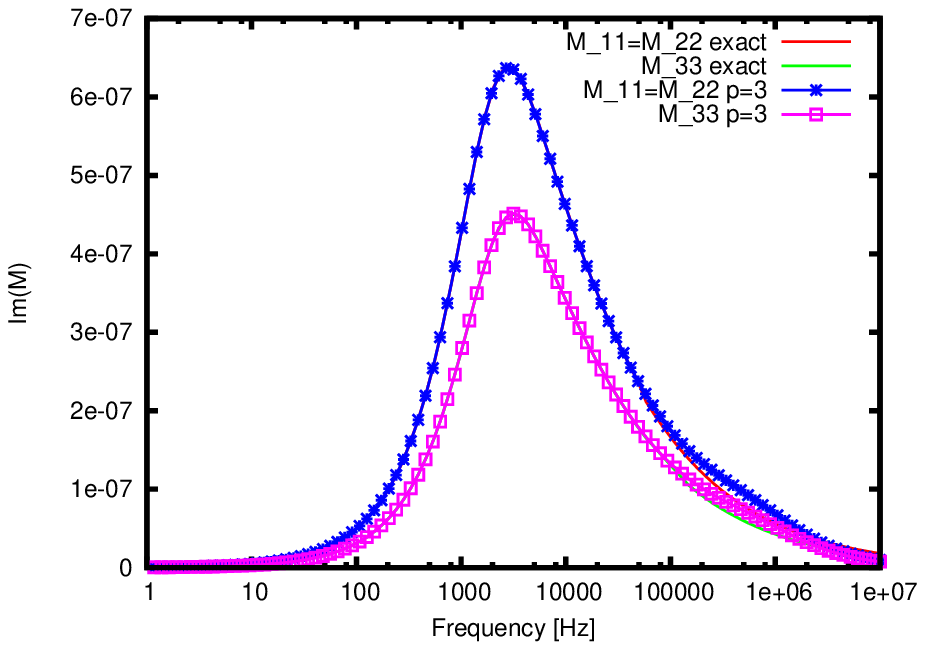} \\
\text{Re}(\widecheck{\widecheck{\mathcal M}}) &  \text{Im}(\widecheck{\widecheck{\mathcal M}})
\end{array}$
\end{center}
\caption{Frequency response of a conducting spheroid with $a=0.02\hbox{m}$, $b=c=0.01\hbox{m}$ with $\sigma_*=1.5\times10^7 \hbox{S/m}$,  $\mu_*=\mu_0$  showing the diagonal coefficients of  $ \widecheck{\widecheck{\mathcal M}} $ based on the analytical solution in \cite{aospheroid}, the numerically computed coefficients using the approach of ~\cite{ledgerlionheart2014} with an unstructured grid of $14\,579$ tetrahedra and uniform $p=3$ elements and the limiting ${\mathcal T}(0)$ coefficients.}
\label{fig:spheroidnonmag}
\end{figure*}
%%%%%%%%%%%%%%%%%%%%%%%%%%%%%%%%%%%%%%%%%%%%%%%%%%%%%%%%%%%%%%%%%%%%%%%%%%%%%%%

Fig.~\ref{fig:spheroidmag} shows the corresponding dependence of the diagonal coefficients of $\widecheck{\widecheck{\mathcal M}} $ with frequency for the same sized spheroid considered in Fig.~\ref{fig:spheroidnonmag}, but now with $\sigma_*=1.5 \times10^7 \hbox{S/m}$ and $\mu_*= 1.5\mu_0$. 
Applying the results of~\cite{schmidt2008}
 we establish that quasi--static model remains valid for this case provided that $f =\omega / ( 2 \pi) \ll 3.6 \text{Mhz}$.

The imaginary parts of $\widecheck{\widecheck{\mathcal M}} $ vanish for low and high frequency limits and the coefficients of the real part tend to the coefficients of $ {\mathcal T} ({\mu}_r )$ and ${\mathcal T} (0)$, respectively, as expected. As with the case shown in Fig.~\ref{fig:spheroidnonmag}, the agreement between the numerically computed tensor coefficients using a mesh of $14\,579$ unstructured tetrahedra and uniform $p=3$ elements and the analytical solution is excellent.

%%%%%%%%%%%%%%%%%%%%%%%%%%%%%%%%%%%%%%%%%%%%%%%%%%%%%%%%%%%%%%%%%%%%%%%%%%%%%%%
\begin{figure*}
\begin{center}
$\begin{array}{cc}
\includegraphics[width=3in]{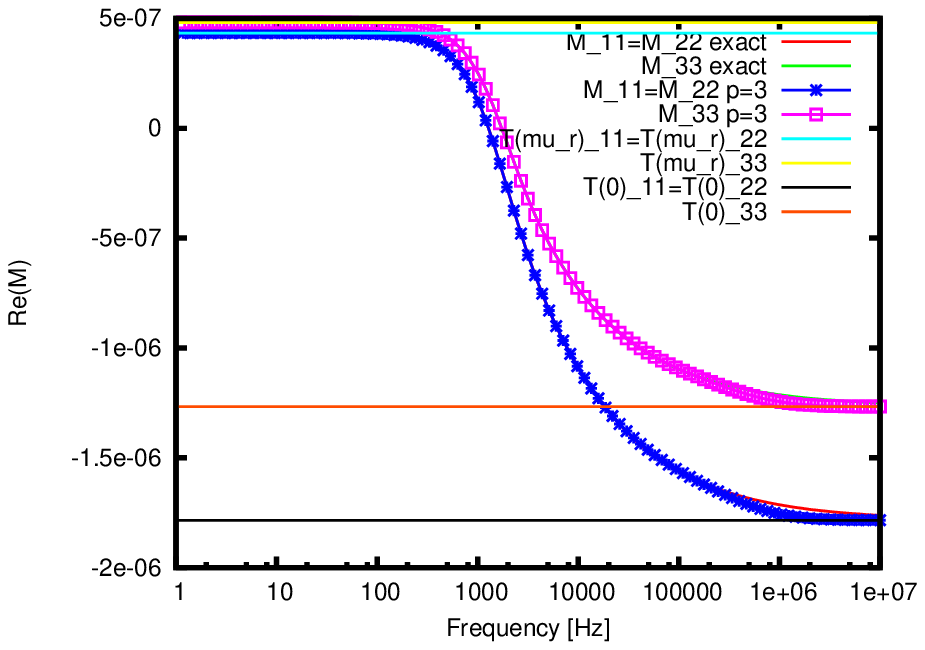} &
\includegraphics[width=3in]{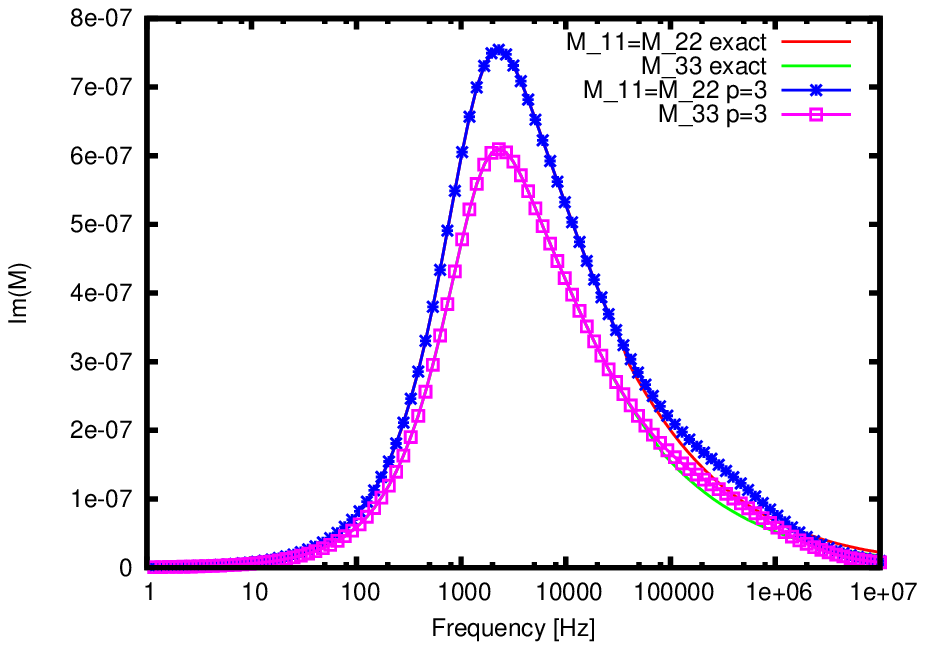} \\
\text{Re}(\widecheck{\widecheck{\mathcal M}}) &  \text{Im}(\widecheck{\widecheck{\mathcal M}})
\end{array}$
\end{center}
\caption{Frequency response of a conducting permeable spheroid with $a=0.02\hbox{m}$, $b=c=0.01\hbox{m}$ with $\sigma_*=1.5\times10^7 \hbox{S/m}$,  $\mu_*=1.5 \mu_0$  showing the diagonal coefficients  of  $\widecheck{\widecheck{\mathcal M}} $ based on the analytical solution in \cite{aospheroid}, the numerically computed coefficients using the approach of ~\cite{ledgerlionheart2014} with an unstructured grid of $14\,579$ tetrahedra and uniform $p=3$ elements and the limiting ${\mathcal T}(\mu_r)$ and ${\mathcal T}(0)$ coefficients.}
\label{fig:spheroidmag}
\end{figure*}
%%%%%%%%%%%%%%%%%%%%%%%%%%%%%%%%%%%%%%%%%%%%%%%%%%%%%%%%%%%%%%%%%%%%%%%%%%%%%%%

\section{Results for multiply connected objects} \label{sect:multfreq}

We first consider a solid torus with major and minor radii, $0.02\, \text{m}$ and $0.01\, \text{m}$, respectively,  and material properties $\sigma_*=5.96 \times10^7 \hbox{S/m}$ and $\mu_*=1.5\mu_0$. In order to investigate the frequency response, we perform a frequency sweep, using the approach of~\cite{ledgerlionheart2014},  and consider the converged non-zero coefficients of $\widecheck{\widecheck{\mathcal M}}$ obtained with $p=2$ on a mesh of $29,  882$ unstructured tetrahedra, which is generated in order to discretise the (unit sized) object and region between the object and a spherical outer boundary with radius $ 100|B|$. The results of the frequency sweep are shown in Figure~\ref{fig:trousmag} where we include, as a comparison, the non--zero coefficients of ${\mathcal T}(\mu_r)$ and ${\mathcal T}(0)$, which have also been computed numerically. 
Note that by computing the shape dependent constants according to~\cite{schmidt2008}, we establish that quasi--static model remains valid for this object provided that $f =\omega / ( 2 \pi) \ll 12.2 \text{Mhz}$.

Despite the fact that this object is multiply connected, with $\beta_0(B)=\beta_1(B)=1$,  $\beta_2(B)=0$, the low frequency coefficients of $\widecheck{\widecheck{\mathcal M}}$ still tend to those of ${\mathcal T}(\mu_r)$, as expected by Theorem~\ref{thm:lowfreq}. However, we do not expect the coefficients of $\widecheck{\widecheck{\mathcal M}}$ to tend to ${\mathcal T}(0)$ as the high limit of quasi-static model is approach, as discussed in Remark 
~\ref{remark:highfreq},
and our numerical experiments confirm that this is indeed a sufficient condition since $\widecheck{\widecheck{\mathcal M}}_{11} \nrightarrow {\mathcal T}(0)_{11}$, although, interestingly, $\widecheck{\widecheck{\mathcal M}}_{22}= \widecheck{\widecheck{\mathcal M}}_{33} \rightarrow {\mathcal T}(0)_{22}= {\mathcal T}(0)_{33}$ indicating a deeper result not covered by the earlier theory.

%%%%%%%%%%%%%%%%%%%%%%%%%%%%%%%%%%%%%%%%%%%%%%%%%%%%%%%%%%%%%%%%%%%%%%%%%%%%%%%
\begin{figure*}
\begin{center}
$\begin{array}{cc}
\includegraphics[width=3in]{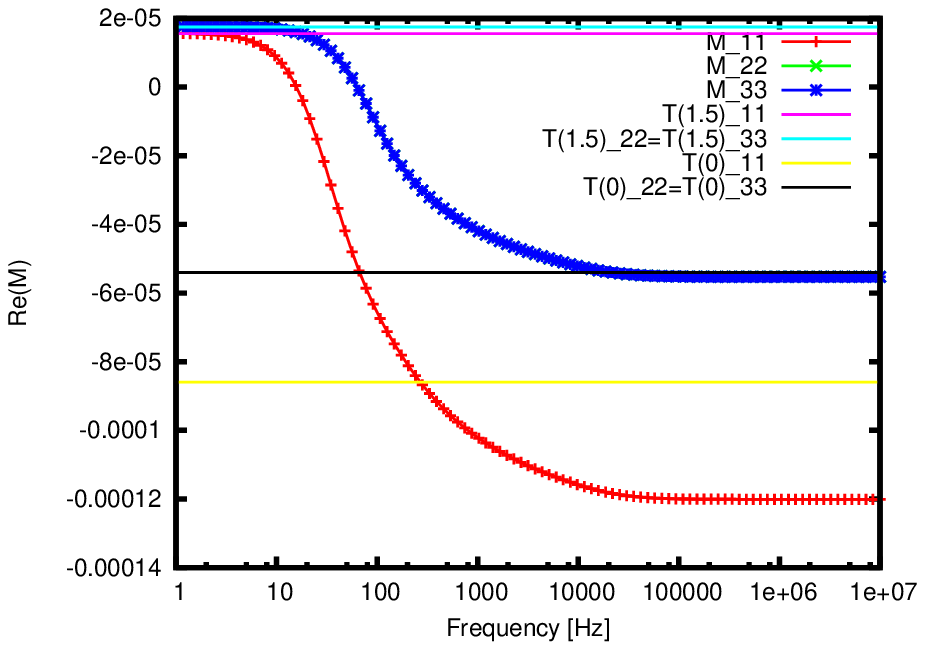} &
\includegraphics[width=3in]{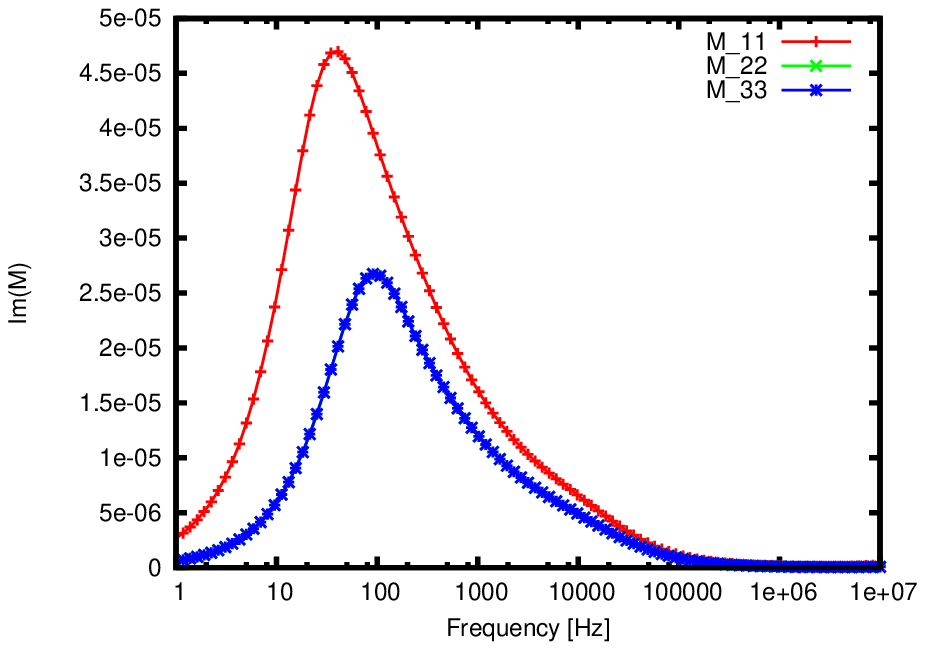} \\
\text{Re}(\widecheck{\widecheck{\mathcal M}}) &  \text{Im}(\widecheck{\widecheck{\mathcal M}})
\end{array}$
\end{center}
\caption{Frequency response of a conducting permeable torus with major and minor radii, $0.02\, \text{m}$ and $0.01\, \text{m}$, respectively,  and material properties $\sigma_*=5.96 \times10^7 \hbox{S/m}$ and $\mu_*=1.5\mu_0$  showing the coefficients of  $\widecheck{\widecheck{\mathcal M}} $ obtained  using the approach of~\cite{ledgerlionheart2014} with an unstructured grid of $29,  882$ tetrahedra and uniform $p=2$ elements and the limiting ${\mathcal T}(\mu_r)$ and ${\mathcal T}(0)$ coefficients.}
\label{fig:trousmag}
\end{figure*}
%%%%%%%%%%%%%%%%%%%%%%%%%%%%%%%%%%%%%%%%%%%%%%%%%%%%%%%%%%%%%%%%%%%%%%%%%%%%%%%

Secondly, we consider a sphere of radius $0.01\,\text{m}$ with a spherical void of $0.005\,\text{m}$ located centrally and the same material parameters as the above torus. We employ an unstructured mesh of $6,873$ tetrahedra with higher order geometry representation and present the converged results obtained for the frequency sweep of the non--zero diagonal coefficients of $\widecheck{\widecheck{\mathcal M}}$ obtained with $p=2$ elements. As in the case of the torus, a spherical outer boundary with radius $100|B|$ is used to truncate the computational domain. The results of the frequency sweep are shown in Figure~\ref{fig:sinsmag} where we include, as a comparison, the non--zero coefficients of ${\mathcal T}(\mu_r)$ and ${\mathcal T}(0)$, which have also been computed numerically. For this object, $\beta_0(B)=\beta_2(B)=1$ and $\beta_1(B)=0$ so that Theorem~\ref{thm:lowfreq} holds in the low frequency case and by Remark 
~\ref{remark:highfreq},
 $\widecheck{\widecheck{\mathcal M}}$ tends to ${\mathcal T}(0)$ at the high frequency, as illustrated in Figure~\ref{fig:sinsmag}. Note that the quasi--static model remains valid for this object provided that $f =\omega / ( 2 \pi) \ll 12.8 \text{Mhz}$.

%%%%%%%%%%%%%%%%%%%%%%%%%%%%%%%%%%%%%%%%%%%%%%%%%%%%%%%%%%%%%%%%%%%%%%%%%%%%%%%
\begin{figure*}
\begin{center}
$\begin{array}{cc}
\includegraphics[width=3in]{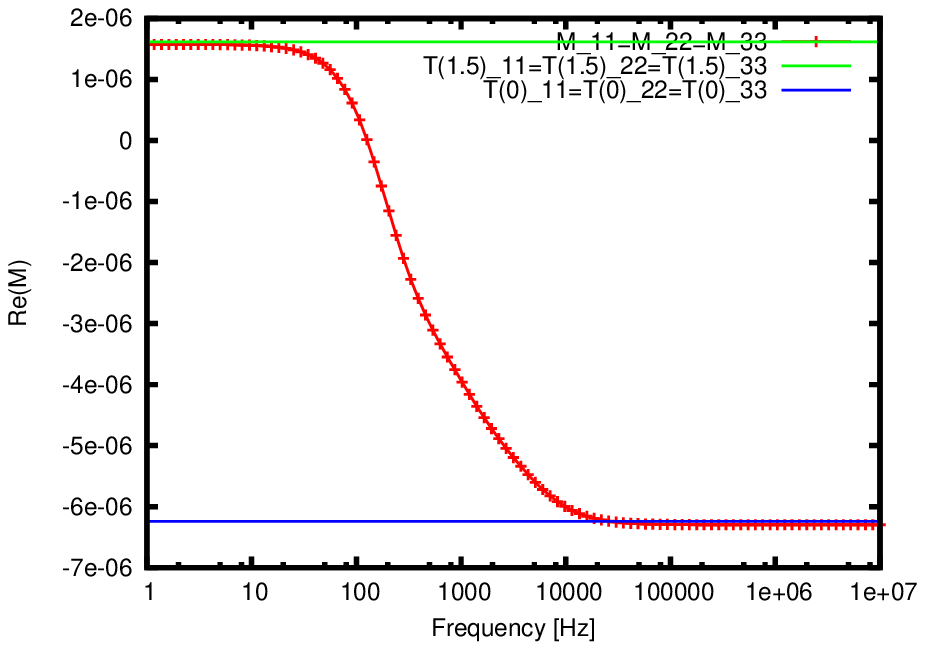} &
\includegraphics[width=3in]{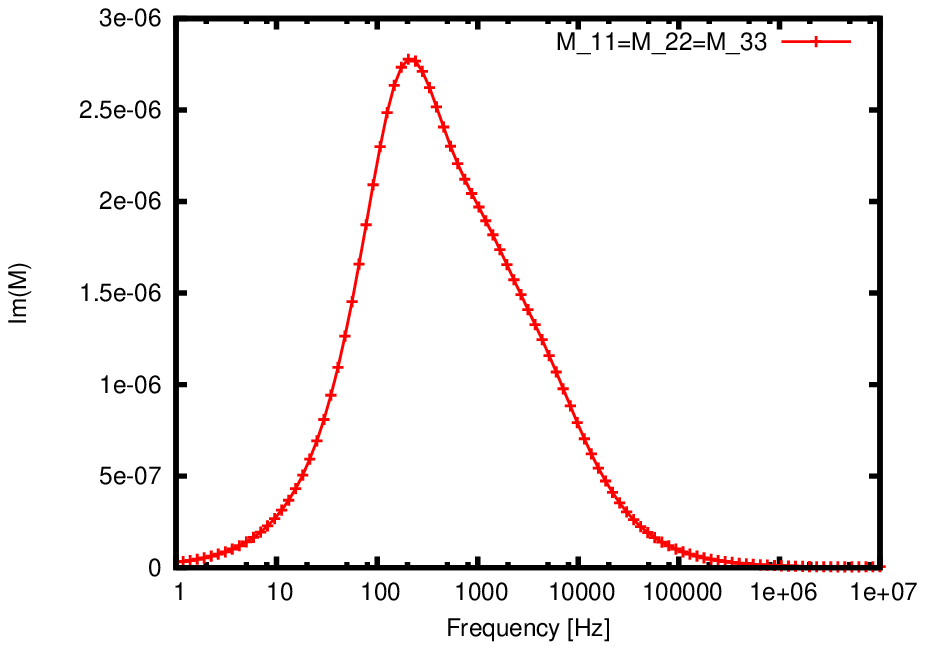} \\
\text{Re}(\widecheck{\widecheck{\mathcal M}}) &  \text{Im}(\widecheck{\widecheck{\mathcal M}})
\end{array}$
\end{center}
\caption{Frequency response of a conducting permeable sphere of raduis $0.01\, \text{m}$  with a centrally placed void of raduis $0.005\, \text{m}$, respectively,  and material properties $\sigma_*=5.96 \times10^7 \hbox{S/m}$ and $\mu_*=1.5\mu_0$  showing the coefficents of   $\widecheck{\widecheck{\mathcal M}} $ obtained  using the approach of~\cite{ledgerlionheart2014} with an unstructured grid of $6,873$ tetrahedra and uniform $p=2$ elements and the limiting ${\mathcal T}(\mu_r)$ and ${\mathcal T}(0)$ coefficients.}
\label{fig:sinsmag}
\end{figure*}
%%%%%%%%%%%%%%%%%%%%%%%%%%%%%%%%%%%%%%%%%%%%%%%%%%%%%%%%%%%%%%%%%%%%%%%%%%%%%%

\section{Results for the Remington rifle cartridge} \label{sect:rem}

Adopting the simplified geometry and the material parameters according to the description of the Remington rifle cartridge given in~\cite{rehimpeyton} a mesh of $23\,551$ unstructured tetrahedra is generated in order to discretise the (unit sized) conducting object  $B$ and the region between the object and a rectangular outer bounding box $(-1000,1000)^3$. 
The Remington rifle cartridge is positioned so that its length is aligned with the $\hat{\vec e}_3$ axis and the cylindrical cross-section lies in the $\hat{\vec e}_1$, $\hat{\vec e}_2$ plane, thus, due to the objects rotational and reflectional symmetries, $\widecheck{\widecheck{\mathcal M}}$ is diagonal~\cite{ledgerlionheart2014} with independent coefficients $ \widecheck{\widecheck{\mathcal M}}_{11}= \widecheck{\widecheck{\mathcal M}}_{22}$ and $ \widecheck{\widecheck{\mathcal M}}_{33}$.
 The convergence of the independent coefficients of the tensor obtained by employing uniform $p=0,1,2$ and $p=3$ elements in turn for a frequency sweep are shown in Fig.~\ref{fig:remingtonshell}. We observe that increasing $p$ yields convergence of the real and imaginary coefficients of $\widecheck{\widecheck{\mathcal M}}$ with the frequency response for $p=2$ and $p=3$ being practically indistinguishable from each other. The curves bear considerable similarity to the frequency response from a non--permeable conducting spheroid shown previously in Fig.~\ref{fig:spheroidnonmag}. Note that by computing the shape dependent constants according to~\cite{schmidt2008}, we establish that quasi--static model remains valid for this object provided that $f =\omega / ( 2 \pi) \ll 3.0 \text{Mhz}$.

%%%%%%%%%%%%%%%%%%%%%%%%%%%%%%%%%%%%%%%%%%%%%%%%%%%%%%%%%%%%%%%%%%%%%%%%%%%%%%%
\begin{figure*}
\begin{center}
$\begin{array}{cc}
\includegraphics[width=3in]{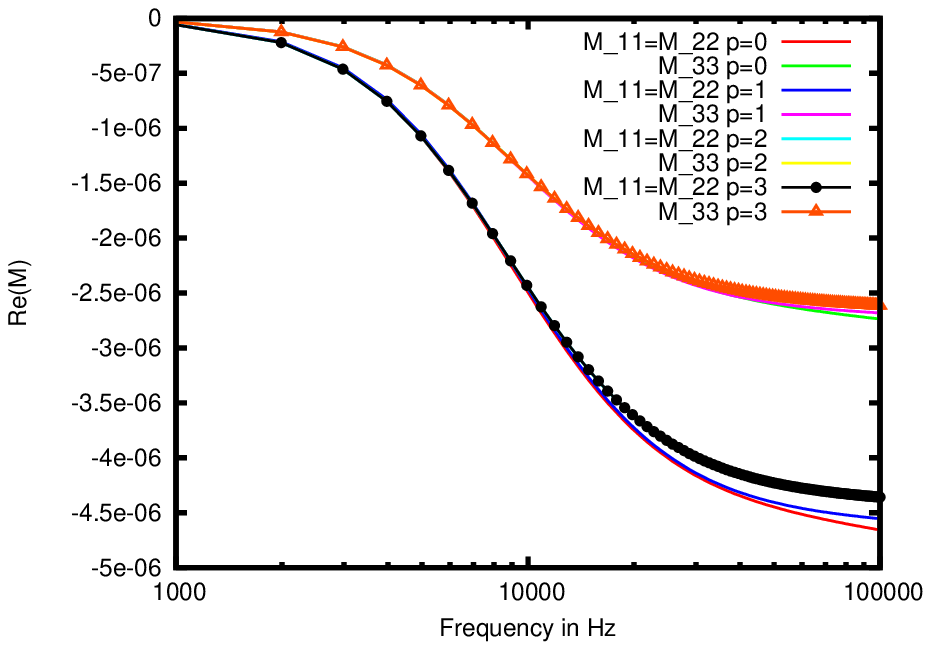} &
\includegraphics[width=3in]{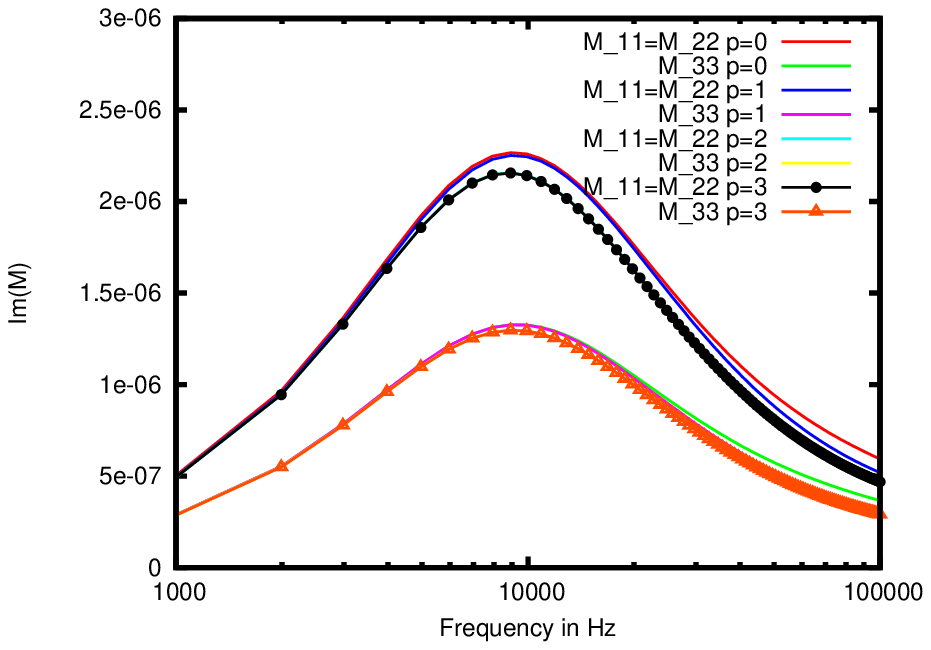} \\
\hbox{Re}(\widecheck{\widecheck{\mathcal M}}) &
 \hbox{Im}(\widecheck{\widecheck{\mathcal M}})
\end{array}$
\end{center}
\caption{Frequency response of a conducting Remington rifle cartridge as defined in~\cite{rehimpeyton} showing the diagonal coefficients of  $ \widecheck{\widecheck{\mathcal M}} $  computed numerically using the approach of ~\cite{ledgerlionheart2014} with an unstructured grid  of $23\,551$ tetrahedra and uniform $p=0,1,2,3$ elements.}
\label{fig:remingtonshell}
\end{figure*}
%%%%%%%%%%%%%%%%%%%%%%%%%%%%%%%%%%%%%%%%%%%%%%%%%%%%%%%%%%%%%%%%%%%%%%%%%%%%%%%

If the Remington rifle cartridge is rotated about an axis then the coefficients of  $\widecheck{\widecheck{\mathcal M}}$  transform according to~\cite{ledgerlionheart2014}
\begin{equation}
\widecheck{\widecheck{\mathcal M}}_{ij}' = {\mathcal R}_{ik} {\mathcal R}_{j \ell } \widecheck{\widecheck{\mathcal M}}_{k\ell} ,
\end{equation}
where ${\mathcal R}$ is the rotation matrix of the transformation. In particular, for a rotation $\theta$  about the $\hat{\vec e}_2$ axis, the components of the transformed tensor are
%%%%%%%%%%%%%%%%%%%%%%%%%%%%%%%%%%%%%%%%%%%%%%%%%%%%%%%%%%%%%%%%%%%%%%%%%%%%%%
\begin{equation}
\widecheck{\widecheck{\mathcal M}}'=\left ( \begin{array}{ccc} \widecheck{\widecheck{\mathcal M}}_{11}\cos^2 \theta + \widecheck{\widecheck{\mathcal M}}_{33}\sin^2 \theta &
0 & \widecheck{\widecheck{\mathcal M}}_{11}\cos \theta \sin \theta - \widecheck{\widecheck{\mathcal M}}_{33}\cos \theta \sin \theta \\
0 & \widecheck{\widecheck{\mathcal M}}_{11} & 0 \\  
\widecheck{\widecheck{\mathcal M}}_{11}\cos \theta \sin \theta - \widecheck{\widecheck{\mathcal M}}_{33}\cos \theta \sin \theta &
0 & \widecheck{\widecheck{\mathcal M}}_{11}\cos^2 \theta + \widecheck{\widecheck{\mathcal M}}_{33}\sin^2 \theta \end{array} \right ) \label{eqn:tranpotens}.
\end{equation}

The frequency response for $4\pi  \widecheck{\widecheck{\mathcal M}}_{33}'$ with rotation through $360$ degrees and for frequencies ranging from $1\hbox{kHz}$ to $100\hbox{kHz}$ is illustrated in Fig.~\ref{fig:remingtonshellrot}. The corresponding rotation response for the same set of frequencies is shown in Fig.~\ref{fig:remingtonshellfreq}. We remark that the magnitude of the real part of the coefficients of $\widecheck{\widecheck{\mathcal M}}$ increases with increasing frequency while the imaginary part of the coefficients peaks at around 10kHz and decays away from smaller and larger frequencies.

%%%%%%%%%%%%%%%%%%%%%%%%%%%%%%%%%%%%%%%%%%%%%%%%%%%%%%%%%%%%%%%%%%%%%%%%%%%%%%%
\begin{figure*}
\begin{center}
$\begin{array}{cc}
\includegraphics[width=3in]{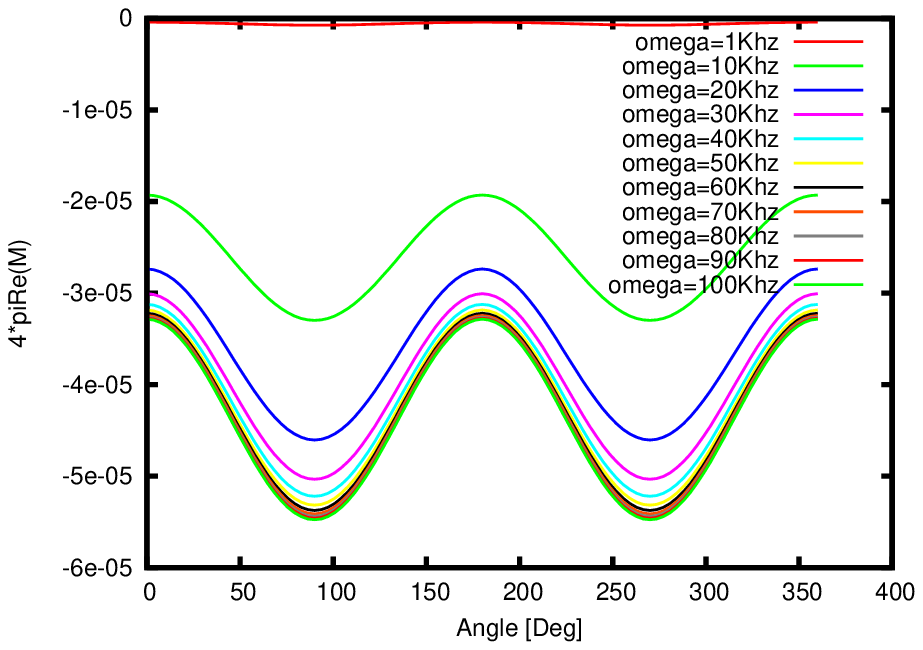} &
\includegraphics[width=3in]{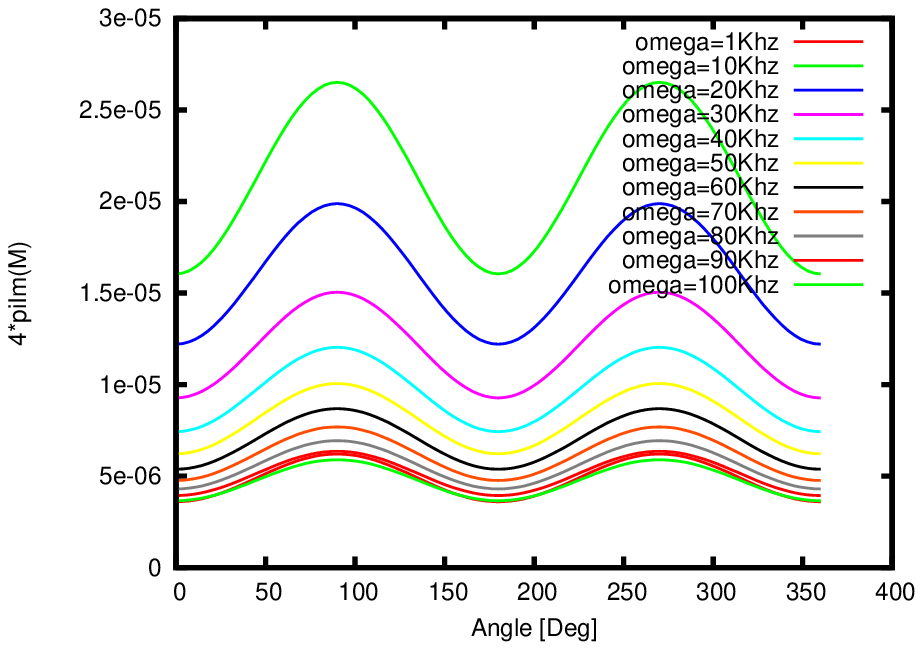} \\
4\pi\hbox{Re}(\widecheck{\widecheck{\mathcal M}}_{33}') & 4\pi\hbox{Im}( \widecheck{\widecheck{\mathcal M}}_{33}')
\end{array}$
\end{center}
\caption{Frequency response of a conducting Remington rifle cartridge as defined in~\cite{rehimpeyton} under rotation showing the transformed  $4\pi \widecheck{\widecheck{\mathcal M}}_{33}'$ for frequencies ranging from  $1\hbox{kHz}$ to $100\hbox{kHz}$  computed numerically using the approach of ~\cite{ledgerlionheart2014} with an unstructured grid  of $23\,551$ tetrahedra and uniform $p=3$ elements.}
\label{fig:remingtonshellrot}
\end{figure*}
%%%%%%%%%%%%%%%%%%%%%%%%%%%%%%%%%%%%%%%%%%%%%%%%%%%%%%%%%%%%%%%%%%%%%%%%%%%%%%%
\begin{figure*}
\begin{center}
$\begin{array}{cc}
\includegraphics[width=3in]{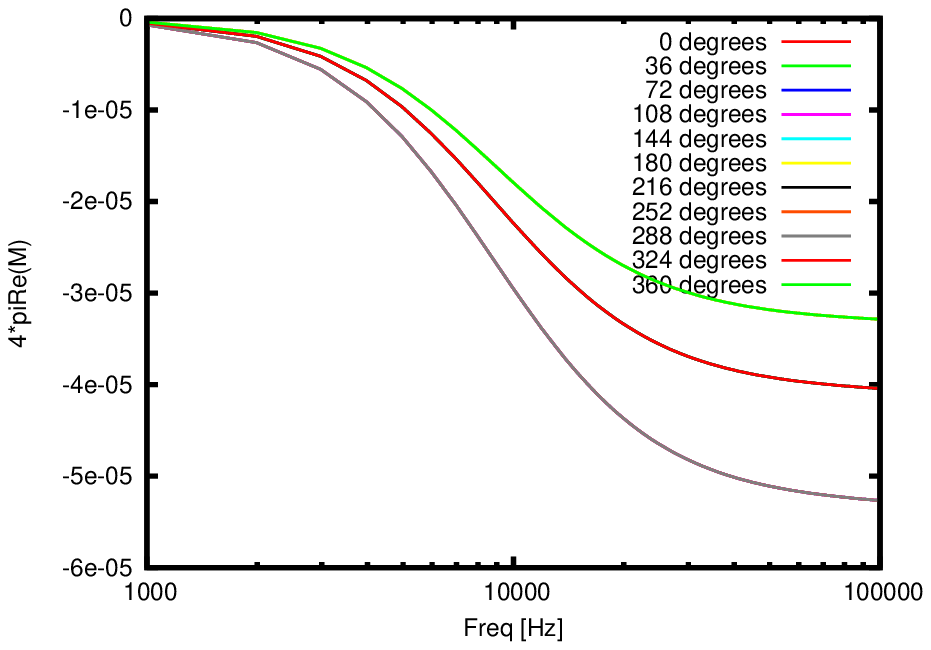} &
\includegraphics[width=3in]{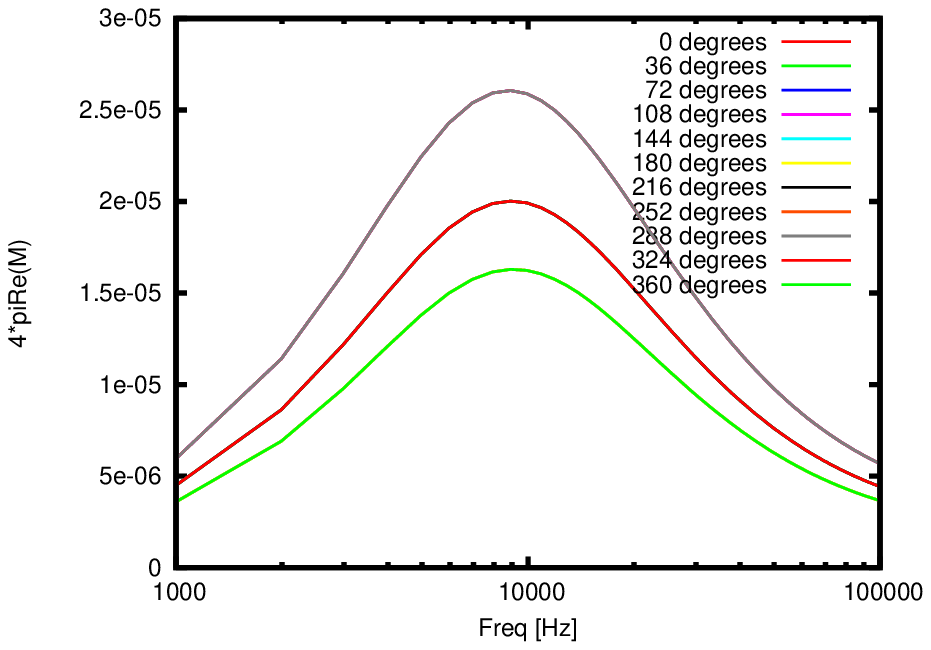} \\
4\pi \hbox{Re}(\widecheck{\widecheck{\mathcal M}}_{33}') & 4\pi \hbox{Im}(\widecheck{\widecheck{\mathcal M}}_{33}')
\end{array}$
\end{center}
\caption{Rotational response of a conducting Remington rifle cartridge as defined in~\cite{rehimpeyton} for different frequencies showing the transformed  $4\pi\widecheck{\widecheck{\mathcal M}}_{33}'$ for rotations ranging from  $0\hbox{ degrees}$ to $360\hbox{ degrees}$  computed numerically using the approach of ~\cite{ledgerlionheart2014} with an unstructured grid  of $23\,551$ tetrahedra and uniform $p=3$ elements.}
\label{fig:remingtonshellfreq}
\end{figure*}
%%%%%%%%%%%%%%%%%%%%%%%%%%%%%%%%%%%%%%%%%%%%%%%%%%%%%%%%%%%%%%%%%%%%%%%%%%%%%%%

The results presented in Figs.~\ref{fig:remingtonshellrot} and~\ref{fig:remingtonshellfreq} closely match the response of the measurements of the coefficients of $\widecheck{\widecheck{\mathcal M}}$ to changes in rotation and changes in frequency presented by
Marsh, Ktisis, J{\"a}rvi, Armitage, and Peyton in~\cite{peyton2013} for the same object.
\section{Conclusion}\label{sect:concl}

The properties of the rank 2 tensor $\widecheck{\widecheck{\mathcal M}}$ and its connection with the P\'oyla--Szeg\"o and the magnetic polarizability tensors has been investigated.
We have described how our results in~\cite{ledgerlionheart2014} provide a framework for the explicit computation of its coefficients.
We have shown, by introducing a splitting of $\widecheck{\widecheck{\mathcal M}}$, that a family of rank 2 tensors can be established, which describe the response from a range of magnetic and conducting objects. Furthermore, bounds on the invariants of the P\'oyla--Szeg\"o tensor have been established and the low frequency and high conductivity limiting cases for the coefficients of $\widecheck{\widecheck{\mathcal M}}$ have been obtained. We have also obtained the behaviour of the coefficients for low conductivity and high frequencies at the limit of applicability of the quasi-static model, which are in agreement with the predictions in Landau and Lifschtiz~\cite{landau}. Interestingly, the connectedness of the object does not play a role in either the low frequency or the low conductivity case, but does in the high frequency and high conductivity cases. The results have been applied ellipsoidal objects, multiply connected objects as well as the frequency and rotational response from a Remington rifle cartridge.

\section*{Acknowledgement}
The authors are very grateful to Professors H. Ammari, J. Chen, A.J. Peyton and D. Volkov for their helpful discussions and comments on polarizability tensors and to
 Dr. K. Schmidt for his help with the calculation of shape dependent constants to deduce the limit of the quasi-static model. They would also like to thank EPSRC for the financial support received from the grants EP/K00428X/1 and EP/K039865/1 and the support received from the land mine research charity Find A Better Way.
\appendix

\section{The engineering connection} \label{sect:reciprocity}
%%%%%%%%%%%%%%%%%%%%%%%%%%%%%%%%%%%%%%%%%%%%%%%%%%%%%%%%%%%%%%%%%%%%%%%%%%%%%%%
In this appendix we provide a connection between the engineering prediction of (\ref{eqn:engexp}) and an asymptotic formula of the form (\ref{eqn:asymformeng}).
Recall the Lorentz reciprocity principal, which is usually formulated for the time harmonic equations, in the form\cite{harrington,landau} 
%%%%%%%%%%%%%%%%%%%%%%%%%%%%%%%%%%%%%%%%%%%%%%%%%%%%%%%%%%%%%%%%%%%%%%%%%%%%%%%
\begin{align}
\nabla \cdot ({\vec E} ^m\times {\vec H}^e - {\vec E}^e \times {\vec H}^m) = {\vec J}_0^m\cdot {\vec E}^e - {\vec J}_0^e \cdot {\vec E}^m,
\end{align}
or, by integrating over ${\mathbb R}^3$ and using the far field behaviour of the fields, as
\begin{align}
\int_{{\mathbb R}^3}  {\vec J}_0^m\cdot {\vec E}^e\dif {\vec x} = \int_{{\mathbb R}^3}  {\vec J}_0^e\cdot {\vec E}^m\dif {\vec x}.\label{eqn:lorentz}
\end{align}
%%%%%%%%%%%%%%%%%%%%%%%%%%%%%%%%%%%%%%%%%%%%%%%%%%%%%%%%%%%%%%%%%%%%%%%%%%%%%%%
It follows from this result that the response is unchanged when the transmitter and receiver are interchanged. Furthermore, if the derivation is repeated for the eddy current model, the result (\ref{eqn:lorentz}) is again obtained. Then, if we follow~\cite{landau}[pg 300], and assume the current sources $m,e$ to have a small support and to be located at ${\vec x}$ and ${\vec y}$, respectively, then the first term in a Taylor series of expansion of the fields ${\vec E}^m$ and ${\vec E}^e$ about the centre of the current source~\footnote{As the size of support of the sources relative to the wavelength and the distance between them tends to zero.} is
%%%%%%%%%%%%%%%%%%%%%%%%%%%%%%%%%%%%%%%%%%%%%%%%%%%%%%%%%%%%%%%%%%%%%%%%%%%%%%%
\begin{align}
{\vec E}^e({\vec x}) \cdot {\vec p}^m \approx {\vec E}^m({\vec y}) \cdot{\vec p}^e,\label{eqn:taylorrecpop1}
\end{align}
%%%%%%%%%%%%%%%%%%%%%%%%%%%%%%%%%%%%%%%%%%%%%%%%%%%%%%%%%%%%%%%%%%%%%%%%%%%%%%%
where ${\vec p}^m$ is the electric dipole moment of the current source $m$.
It is important to note that this is only the first term in the Taylor's series expansion, including the next term leads to
%%%%%%%%%%%%%%%%%%%%%%%%%%%%%%%%%%%%%%%%%%%%%%%%%%%%%%%%%%%%%%%%%%%%%%%%%%%%%%%
\begin{align}
{\vec E}^e({\vec x}) \cdot {\vec p}^m&  +2\nabla^s{\vec E}^e({\vec x}): {\mathcal R}^m +
{\vec B}^e({\vec x}) \cdot {\vec m}^m \approx \nonumber \\
&{\vec E}^m({\vec y}) \cdot {\vec p}^e +2 \nabla^s{\vec E}^m({\vec y}): {\mathcal R}^e +
{\vec B}^m({\vec y}) \cdot {\vec m}^e \label{eqn:taylorrecpop2},
\end{align}
%%%%%%%%%%%%%%%%%%%%%%%%%%%%%%%%%%%%%%%%%%%%%%%%%%%%%%%%%%%%%%%%%%%%%%%%%%%%%%%
where ${\mathcal R}^m$ is a quadrupole moment of the current source $m$, ${\vec m}^m$ the magnetic moment of the same current source~\cite{landau} and exact reciprocity is expected if all the terms in the Taylor series expansion are considered. 

For the eddy current problem described in this work and coils located in free space that can be idealised as dipoles with a magnetic moment, only,  reciprocity implies that $ {\vec m}^m \cdot {\vec H}_\alpha^e({\vec x}) \approx {\vec m}^e\cdot  {\vec H}_\alpha^m({\vec y}) $ i.e. the result is the same if ${\vec x}$ and ${\vec w}$ and ${\vec m}^m$ and ${\vec m}^e$ are interchanged. Considering (\ref{eqn:asympformula}), we have in vector notation,
%%%%%%%%%%%%%%%%%%%%%%%%%%%%%%%%%%%%%%%%%%%%%%%%%%%%%%%%%%%%%%%%%%%%%%%%%%%%%%%
\begin{align}
{\vec m}^m \cdot  {\vec H}_\alpha^e({\vec x}) - {\vec m}^e \cdot  {\vec H}_\alpha^m({\vec y}) = &
  {\vec m}^m  \cdot ( {\vec D}^2G({\vec x},{\vec z})  {\mathcal A}
({\vec D}^2G({\vec z},{\vec y})  {\vec m}^e )) \nonumber\\
& -  {\vec m}^e  \cdot ({\vec D}^2G({\vec y},{\vec z})  {\mathcal A} 
({\vec D}^2G({\vec z},{\vec x}) {\vec m}^m)) 
+   \Delta ({\vec m}^m , {\vec R}^e) - \Delta ({\vec m}^e, {\vec R}^m) 
\label{eqn:recpasympexpand},
\end{align}
%%%%%%%%%%%%%%%%%%%%%%%%%%%%%%%%%%%%%%%%%%%%%%%%%%%%%%%%%%%%%%%%%%%%%%%%%%%%%%%
where $\Delta ({\vec m} , {\vec R}):= {\vec m} \cdot   {\vec R} $ and, in the case considered,  ${\vec H}_0^e({\vec x})={\vec D}^2G({\vec x},{\vec y}) {\vec m}^e$ and ${\vec H}_0^m({\vec w})={\vec D}^2 G({\vec y},{\vec x}) {\vec m}^m$, thus, from the symmetry of ${\vec D}^2G({\vec x},{\vec y})$, we have used ${\vec m}^m \cdot {\vec H}_0^e({\vec x})={\vec m}^e \cdot {\vec H}_0^m({\vec y})$. It follows from (\ref{eqn:mainresultb}),~\cite{ledgerlionheart2014,ammarivolkov2013} that
%%%%%%%%%%%%%%%%%%%%%%%%%%%%%%%%%%%%%%%%%%%%%%%%%%%%%%%%%%%%%%%%%%%%%%%%%%%%%%%
\begin{align}
|  \Delta ({\vec m}^m ,  {\vec R}^e)  - \Delta ({\vec m}^e, {\vec R}^m) |  \le &  | {\vec m}^m|     | {\vec R}^e( {\vec x} )  | +  | {\vec m}^e |  | {\vec R}^m ( {\vec y} ) |  \nonumber \\
 \le& C \alpha^4 \left (  | {\vec m}^m |   \| {\vec H}_0^e \|_{W^{2,\infty} (B_\alpha)} + | {\vec m}^e |    \| {\vec H}_0^m \|_{W^{2,\infty}   ( B_\alpha)} \right ) \nonumber,
\end{align}
%%%%%%%%%%%%%%%%%%%%%%%%%%%%%%%%%%%%%%%%%%%%%%%%%%%%%%%%%%%%%%%%%%%%%%%%%%%%%%%
thus  (\ref{eqn:recpasympexpand}) is an asymptotic expansion for ${\vec m}^m \cdot {\vec H}_\alpha^e({\vec x}) -{\vec m}^e \cdot {\vec H}_\alpha^m({\vec y})$ as $\alpha \to 0$ with  $ \Delta ({\vec m}^m , {\vec R}^e) - \Delta ({\vec m}^e, {\vec R}^m)= O(\alpha^4)$. Then, upto this residual term, 
%%%%%%%%%%%%%%%%%%%%%%%%%%%%%%%%%%%%%%%%%%%%%%%%%%%%%%%%%%%%%%%%%%%%%%%%%%%%%%%
\begin{align}
{\vec m}^m & \cdot ({\vec D}^2G({\vec x},{\vec z})  {\mathcal A}    ( {\vec D}^2G({\vec z},{\vec y}) {\vec m}^e) ) \approx {\vec m}^e \cdot  ( {\vec D}^2 G({\vec y},{\vec z})  {\mathcal A}   ( {\vec D}^2G({\vec z},{\vec x})  {\vec m}^m) ) . \label{eqn:recopgeneral}
\end{align}
%%%%%%%%%%%%%%%%%%%%%%%%%%%%%%%%%%%%%%%%%%%%%%%%%%%%%%%%%%%%%%%%%%%%%%%%%%%%%%%
In light of (\ref{eqn:lorentz}), if one constructs a suitable ${\vec J}_0^m$, which has non-zero support on the measurement coil and is such that the  resulting field ${\vec H}_0^m$ can be idealised as a magnetic dipole,  the induced voltage, $V^{ind}$, as a result of the perturbation caused by the presence of  a general conducting  object, is
%%%%%%%%%%%%%%%%%%%%%%%%%%%%%%%%%%%%%%%%%%%%%%%%%%%%%%%%%%%%%%%%%%%%%%%%%%%%%%
\begin{align}
V^{ind} \approx & C  {\vec m}^m \cdot   ( {\vec D}^2G({\vec x},{\vec z}) {\mathcal A}   ({\vec D}^2G({\vec z},{\vec y})  {\vec m}^e))  \approx  C  {\vec H}_0^m({\vec z}) \cdot ( {\mathcal A} {\vec H}_0^e({\vec z}) ).
\end{align}

up to a scaling constant $C$.

\bibliographystyle{IEEEtranS}
\bibliography{IEEEabrv,ledgerlionheart}

\end{document}